\documentclass[journal,a4,twocolumn,twoside]{IEEEtran}

\usepackage{cite}
\usepackage{subfigure}
\usepackage{amssymb,amsmath}
\usepackage{algorithmic}

\usepackage{color}
\usepackage{epstopdf}
\usepackage{multirow}
\usepackage{multicol}
\usepackage{cases}
\usepackage{setspace}
\usepackage{amsthm}
\usepackage{caption}
\usepackage{delarray}
\usepackage{tikz}
\usepackage[ruled]{algorithm2e}
\SetKwInput{KwInput}{Input}
\SetKwInput{KwOutput}{Output}
\SetKwInput{Initialization}{Initialization}
\usepackage{balance}
\usepackage[mathscr]{euscript}
 \let\mathscr\relax
\usepackage[scr]{rsfso}
\allowdisplaybreaks
\usepackage{bbm}

\newtheorem{theorem}{\bf {Theorem}}

\newtheorem{remark}{{\bf{Remark}}}

\newtheorem{lemma}{\bf {Lemma}}

\newcommand{\st}{{\mathrm{s.t.}}}

\renewcommand{\algorithmicrequire}{\textbf{Input:}}

     \usepackage[compact]{titlesec}
    \titlespacing{\section}{0pt}{2ex}{1ex}
    \titlespacing{\subsection}{0pt}{1ex}{0ex}
    \titlespacing{\subsubsection}{0pt}{0.5ex}{0ex}

\usepackage{capt-of}
\usepackage{dblfloatfix}
\usepackage{varwidth}
\usepackage{tikz}

\usepackage{graphicx}
\usepackage{epstopdf}
\DeclareGraphicsExtensions{.eps,.pdf} 
\graphicspath{ {figures/} }

\usepackage{setspace}
\renewcommand{\baselinestretch}{1} 

\usepackage{amsfonts}
\usepackage{geometry}
\usepackage{array}
\usepackage{booktabs}
\usepackage{siunitx}

\usepackage{babel}
\usepackage{hyperref}

\usepackage{hyperref}
\hypersetup{
 unicode=false,          
 pdftoolbar=true,        
 pdfmenubar=true,        
 pdffitwindow=false,     
 pdfstartview={FitH},    
 pdftitle={My title},    
 pdfauthor={Author},     
 pdfsubject={Subject},   
 pdfcreator={Creator},   
 pdfproducer={Producer}, 
 pdfkeywords={keyword1, key2, key3}, 
 pdfnewwindow=true,      
 colorlinks=true,       
 linkcolor=purple,          
 citecolor=purple,        
 filecolor=purple,      
 urlcolor=purple           
}

\makeatother

\begin{document}
        \setlength{\abovedisplayskip}{4.5pt}
	\setlength{\belowdisplayskip}{4.5pt}
\title{\huge Performance Analysis and Sensing-Aware Resource Allocation for Full-Duplex Massive MIMO ISCC Networks}
\author{
\IEEEauthorblockN{Vo Phi Son, Van-Dinh Nguyen, Nhan Thanh Nguyen, Tuan-Vu Truong and Symeon Shatzinotas\\}
\thanks{V. P. Son and T.-V. Truong are with Smart Green Transformation Center and College of Engineering and Computer Science, VinUniversity, Vietnam (e-mail: \{son.vp, vu.tt\}@vinuni.edu.vn). N. T. Nguyen is with Centre for Wireless Communications, University of Oulu, Finland (e-mail: nhan.nguyen@oulu.fi). V.-D. Nguyen (corresponding author) is with Trinity College Dublin,  Ireland  (e-mail: dinh.nguyen@tcd.ie). S. Chatzinotas is with SnT, University of Luxembourg, Luxembourg (e-mail: symeon.chatzinotas@uni.lu).}
\vspace{-5pt}}

\markboth{Preprint}%
{Vo \MakeLowercase{\textit{et al.}}: Performance Analysis and Sensing-Aware Resource Allocation for FD mMIMO ISCC Networks}
\maketitle 

\begin{abstract}
Emerging edge applications must sense their environment and process data within tight latency budgets. Full-duplex (FD) integrated sensing, communication, and computing (ISCC) supports these functions concurrently, but radar probing interferes with uplink reception while user transmissions degrade target estimation. These coupled effects call for joint sensing and offloading decisions. We investigate a FD massive multiple-input multiple-output (mMIMO) ISCC network in which users partially offload tasks to an edge server while an access point probes a target. For zero-forcing uplink reception with imperfect channel estimates, we derive a tractable rate bound and target-angle Cramér–Rao lower bounds (CRLBs) under the stated interference and beam-alignment assumptions. We then jointly optimize offloading fractions, user and radar transmit powers, and edge computing resources to minimize average end-to-end task latency subject to sensing-accuracy, communication-rate, energy, deadline, and resource constraints. A successive convex approximation algorithm is developed to efficiently address the resulting nonconvex problem. The analysis shows that larger arrays improve sensing and uplink performance, but residual radar interference creates a finite rate ceiling when user powers scale down with array size. Simulation results show joint allocation reduces radar transmit power by approximately $86\%$ relative to equal communication-user power at a $-40$ dB CRLB threshold and lowers latency compared with fixed-allocation and half-duplex schemes.
\end{abstract}

\section{Introduction}\label{sec_Intro}
Integrated sensing, communication, and computation (ISCC) supports emerging edge applications, such as autonomous driving and smart cities, that require timely sensing and data processing under limited spectrum and energy resources \cite{Wen_ISCC_survey}. By sharing communication, radar, and computing resources, ISCC enables joint decisions across these functions \cite{Qi_ISCC_2022}. Full-duplex (FD) massive multiple-input multiple-output (mMIMO) systems can receive uplink data while transmitting radar signals \cite{Zhang_MIMO_2021}. However, radar probing interferes with data reception, while user transmissions degrade target estimation. These effects couple sensing accuracy with task-offloading latency and motivate joint allocation of transmission and computing resources.

\subsection{Related Work}\label{sec_related_work}
Research on integrated sensing and communication (ISAC) provides the basis for coordinating radar sensing and data transmission in ISCC. Radar-centric schemes embed data in radar waveforms or sidelobes through modulation and waveform design \cite{Ma_JSAC_2021, Ma_ISAC_2021, Huang_ISAC_2020, Hassanien_2016}. Communication-centric schemes reuse communication signals to estimate target parameters \cite{Kumari_ISAC_2018}. For FD-ISAC, Liu \textit{et al.} \cite{Liu_2023} jointly designed transmit and receive beamforming to manage communication and sensing interference. The work in \cite{Jiang_ISAC_2026} considered FD-ISAC-assisted device-to-device communication to improve spectral efficiency. Half-duplex (HD) operation separates communication and sensing in time, reducing their mutual interference but limiting concurrent operation. FD avoids this separation at the cost of residual self-interference (RSI), which must be considered when balancing sensing accuracy and communication rates.

When computation offloading is added to these sensing and communication functions, this interference trade-off also affects task delay, energy consumption, and processing-resource allocation. Huang \textit{et al.} \cite{Huang_ISAC_2024} used an edge server to support short-packet offloading and minimized energy consumption under radar constraints. The authors of \cite{Huynh_ISSC_2024} studied ISCC for low-latency mobile edge computing (MEC) services, jointly addressing end-to-end (e2e) latency and service installation costs. Chen \textit{et al.} \cite{Chen_ISCC_2024} considered secure MEC assisted by an FD-ISAC station, using radar transmission for sensing and jamming while allocating computing resources. Qi \textit{et al.} \cite{Qi_ISCC_2022} developed a joint sensing, computing, and communication framework in which sensors perform local processing before transmission. These studies establish the importance of coordinating offloading with sensing. In the uplink setting considered here, conventional communication users and MEC users also share the receive array. Their different service requirements and mutual interference introduce additional coupling between power allocation and task completion.

Managing this coupling requires effective spatial interference suppression, motivating the use of massive arrays. These arrays extend conventional MIMO ISAC designs \cite{Liu_ISAC_2018,Li_ISAC_MIMO_2016,Liu_ISAC_MIMO_2016} by providing greater spatial gain and interference suppression \cite{Buzzi_ISAC_MIMO_2019,Temiz_ISAC_MIMO_2022}. The beamforming designs in \cite{Nhan_2024,Nhan_2023} improve communication rates under sensing constraints. Analytical rate and sensing expressions allow resource allocation without repeatedly averaging over small-scale channel realizations. Liao \textit{et al.} \cite{Liao_ISAC_2024} derived communication signal-to-interference-plus-noise ratio (SINR) and sensing mainlobe-to-average-sidelobe metrics for massive MIMO ISAC power allocation. Topal \textit{et al.} \cite{Topal_2025} considered multi-target angle estimation through Cram\'er--Rao lower bounds (CRLBs), while \cite{Huy_ISAC_2026} used communication and sensing metrics for energy-efficient design. These results motivate a corresponding analysis of uplink offloading in FD-mMIMO ISCC, where radar leakage affects data reception and uplink signals enter the sensing interference.

\subsection{Motivation and Main Contributions} \label{sec_motivation_and_contributions}
Previous studies have investigated communication--sensing trade-offs in FD systems~\cite{Liu_2023} and joint transmission and computing resource allocation for MEC-assisted sensing~\cite{Huang_ISAC_2024, Huynh_ISSC_2024}. These studies motivate us to examine how FD-mMIMO networks can reduce task latency while meeting target-angle accuracy requirements. In our setting, conventional users must meet minimum communication rates, whereas MEC users must complete tasks within latency and energy budgets. As both groups interfere with sensing, increasing radar power can improve target estimation but reduce offloading rates, while increasing user power can require more radar power to maintain sensing accuracy. Since task completion also depends on execution time, power allocation must be considered together with offloading fractions and edge computing resources~\cite{Huang_ISAC_2024,Huynh_ISSC_2024}.

To study these coupled decisions, we consider partial offloading in a monostatic FD mMIMO network operating in target-tracking mode, where a prior target-angle estimate determines the sensing direction~\cite{Liu2021CRB}. We derive uplink rate and angular estimation bounds that describe how user and radar powers affect communication and sensing performance under residual self-interference and imperfect channel state information (CSI). The analysis also identifies the effects of array size. We then use these expressions to jointly allocate wireless and computing resources based on large-scale channel information, subject to the requirements of communication and MEC users. The main contributions are summarized as follows:
\begin{itemize}
\item We derive an uplink rate bound for zero-forcing (ZF) reception with imperfect CSI, accounting for channel-estimation errors, RSI, and target-echo interference. We also derive angular CRLBs under the stated interference and beam-alignment assumptions. These expressions show how radar and user transmit powers affect both offloading rates and target estimation accuracy.

\item We analyze how antenna numbers and transmit powers affect communication and sensing performance. The uplink rate increases with the receive-array size when user powers remain fixed, but approaches a finite limit when all user powers scale as $1/N_{\mathtt{r}}$, with fixed radar power and transmit-array size. The sensing analysis identifies the different roles of transmit-array gain and receive-array resolution in improving angular accuracy.

\item We formulate a joint resource allocation problem to minimize average e2e task latency by optimizing offloading fractions, user and radar powers, and edge CPU resources. The formulation accounts for task deadlines, energy budgets, minimum communication rates, and sensing requirements. We develop a successive convex approximation (SCA)-based algorithm that approximates the coupled rate, delay, and energy terms through convex subproblems and expresses the sensing constraints in affine form.

\item We compare the analytical results with Monte Carlo simulations and examine the algorithm's convergence and performance under different sensing thresholds, antenna configurations, task sizes, and user loads. Comparisons with fixed-allocation and HD schemes show how jointly adapting transmit powers and offloading fractions reduces task latency in the considered scenarios.
\end{itemize}

\subsection{Paper Organization and Notations}\label{sec_paper_structure_notations}
Section~\ref{sec_SystemModel} presents the system model, and Section~\ref{sec_System_analysis} derives the uplink rate and sensing CRLBs. Section~\ref{sec_Problem_Proposed_solution} develops the optimization problem and solution, followed by numerical results in Section~\ref{sec_NumericalResults} and conclusions in Section~\ref{sec_conclusion}.

\textbf{Notation:} Bold lowercase and uppercase letters denote vectors and matrices, respectively. The operators $(\cdot)^*$, $(\cdot)^{\mathtt{T}}$, $(\cdot)^{\mathtt{H}}$, and $\operatorname{Tr}(\cdot)$ denote conjugation, transpose, Hermitian transpose, and trace, respectively. The symbols $|\cdot|$, $\|\cdot\|$, and $\mathbb{E}[\cdot]$ denote absolute value, Euclidean norm, and expectation, respectively. We write $\dot y_x=\partial y/\partial x$, and use $\Re(\cdot)$ and $\Im(\cdot)$ for real and imaginary parts. The notation $\mathcal{CN}(\boldsymbol{\mu},\boldsymbol{R})$ denotes a circularly symmetric complex Gaussian distribution with mean $\boldsymbol{\mu}$ and covariance $\boldsymbol{R}$.

\section{System Model}\label{sec_SystemModel}
\begin{figure}[t]
	\centering
	\includegraphics[width=0.92\columnwidth,trim={0cm 0.0cm 0cm 0.0cm}]{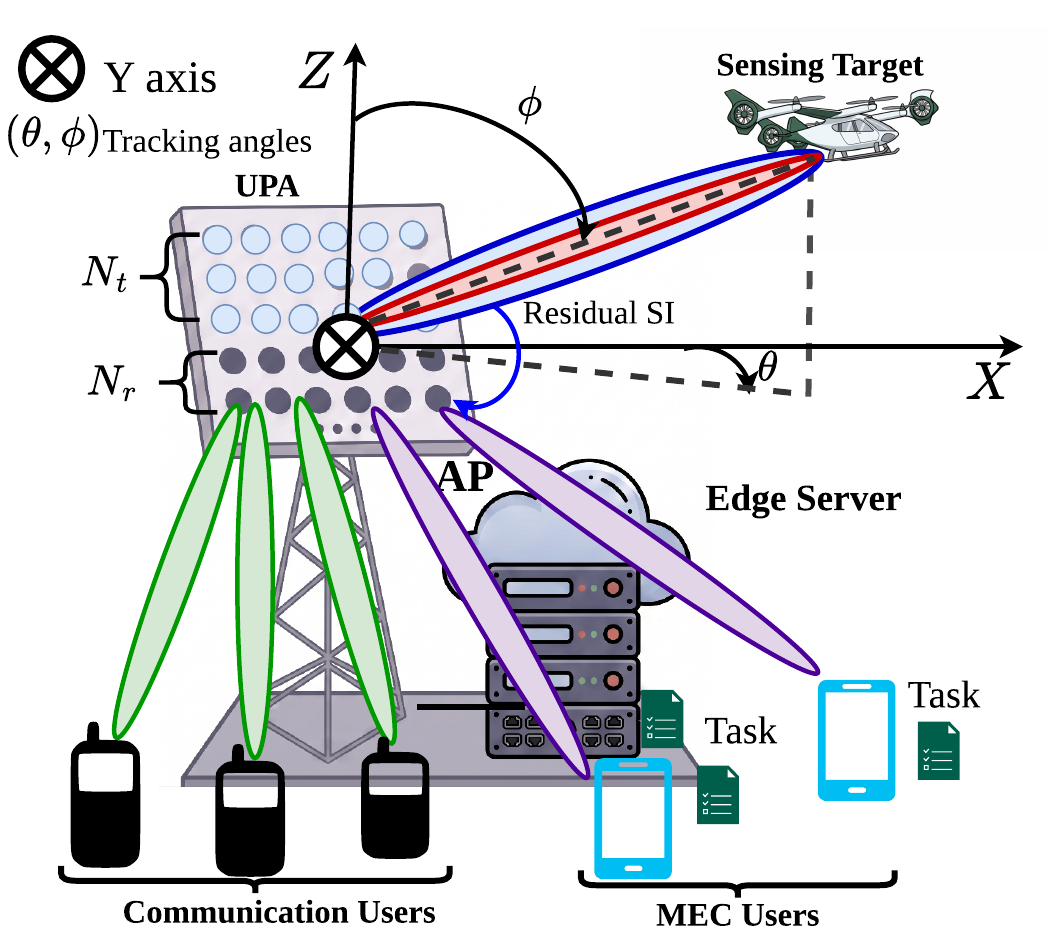}
	\caption{Illustration of monostatic FD-mMIMO ISCC system.}
	\label{fig:system_model}
\end{figure}
As illustrated in Fig.~\ref{fig:system_model}, we consider a monostatic FD-mMIMO ISCC network comprising a single access point (AP) integrated with an edge server (ES), a sensing target, and $K=K_{\mathtt{MEC}}+K_{\mathtt{com}}$ single-antenna user equipment (UE) devices, where $K_{\mathtt{MEC}}$ and $K_{\mathtt{com}}$ denote the numbers of priority MEC UEs requiring high uplink rates to offload tasks to an ES and conventional communication UEs, respectively. UEs $k\in\mathcal{K}_{\mathtt{MEC}}$ require computation offloading, whereas UEs $k'\in\mathcal{K}_{\mathtt{com}}$ engage in conventional uplink communications, with $\mathcal{K}=\mathcal{K}_{\mathtt{MEC}}\cup\mathcal{K}_{\mathtt{com}}$. The AP is equipped with uniform planar arrays (UPAs) comprising $N_{\mathtt{t}}$ transmit and $N_{\mathtt{r}}$ receive antennas, where $N_{\mathtt{t}},N_{\mathtt{r}}\gg 1$.

The system operates in three phases. In phase 1, the UEs transmit orthogonal pilots for channel estimation. In phase 2, the UEs transmit offloaded tasks and communication data to the AP, while the AP simultaneously transmits a radar probing signal toward the target located at a specific angle, which has been determined during the radar scanning or searching phase. The task of user $k\in \mathcal{K}_{\mathtt{MEC}}$ is modeled as $J_k \triangleq (D_k, C_k, T_k^{\max})$, where $D_k$, $C_k$, $T_k^{\max}$ denote the task size, required CPU cycles, and task deadline, respectively. We adopt a sensing-aware partial offloading strategy. We define $\boldsymbol{\varphi} \triangleq \{\varphi_k \, | 0\leq  \varphi_k \leq 1\}_{\forall k}$ as the vector of task fractions offloaded to the ES, while $1-\varphi_k$ denotes the remaining task fraction executed locally at the user's device. In phase 3, the ES executes the offloaded tasks. Note that local task processing at the UEs occurs in parallel across all phases.
\begin{remark}
Conventional FD-ISAC architectures suffer from severe uplink-to-downlink interference~\cite{Tang_2026, Liu_2023}. Exploiting asymmetric MEC traffic, we defer downlink to Phase 3. Reserving Phase 2 for uplink transmission and sensing avoids simultaneous uplink and downlink data transmission.
\end{remark}

\subsection{Communication Model}\label{sec_Communications_model}
\noindent\textit{1) Uplink pilot training and imperfect CSI:} During the pilot training phase, each user $u \in \mathcal{K}$ transmits a pilot sequence of length $\tau_p \ge K$ with duration $t_p = \tau_p T_s$, where $T_s$ denotes the symbol period. The uplink channel is modeled as $\boldsymbol{h}_u = \beta_u^{1/2} \bar{\boldsymbol{h}}_u \in \mathbb{C}^{N_{\mathtt{r}} \times 1}$, where $\beta_u$ accounts for large-scale fading and $\bar{\boldsymbol{h}}_u \sim \mathcal{CN}(\boldsymbol{0}, \boldsymbol{I}_{N_{\mathtt{r}}})$ represents small-scale Rayleigh fading.  Let $\hat{\boldsymbol{h}}_u$ and $\boldsymbol{e}_u$ represent the channel estimate and the corresponding channel-estimation error, such that $\boldsymbol{e}_u = \boldsymbol{h}_u - \hat{\boldsymbol{h}}_u$. Assuming minimum mean square error (MMSE) estimation, $\hat{\boldsymbol{h}}_u$ and $\boldsymbol{e}_u$ are independent. Additionally, assuming all UEs transmit pilots with constant power $\vartheta_u$, $\hat{\boldsymbol{h}}_u \sim \mathcal{CN}(\boldsymbol{0}, \xi_u \boldsymbol{I}_{N_{\mathtt{r}}})$ and $\boldsymbol{e}_u \sim \mathcal{CN}(\boldsymbol{0}, \epsilon_u \boldsymbol{I}_{N_{\mathtt{r}}})$ where the estimation variance $\xi_u$ and error variance $\epsilon_u$ are given by~\cite{Nhan_2025}
\begin{align}
\xi_u &= \frac{\tau_p \vartheta_u \beta_u^2}{\tau_p \vartheta_u \beta_u + \sigma_c^2}, \quad \epsilon_u = \frac{\beta_u \sigma_c^2}{\tau_p \vartheta_u \beta_u + \sigma_c^2}, \quad \forall u \in \mathcal{K}.
\end{align}

\noindent\textit{2) FD-mMIMO ISCC communication signal model:} In phase $2$, MEC UEs ($k \in \mathcal{K}_{\mathtt{MEC}}$) offload task fractions $\varphi_k D_k$ by transmitting $x_k = \sqrt{p_k}s_k$, while conventional communication UEs ($k' \in \mathcal{K}_{\mathtt{com}}$) transmit $x_{k'} = \sqrt{p_{k'}}s_{k'}$. Here, $s_k, s_{k'} \sim \mathcal{CN}(0,1)$ represent normalized information symbols, with $p_k$ and $p_{k'}$ being the respective transmit powers. Concurrently, the AP transmits the sensing signal {$\boldsymbol{x}_{\mathtt{rad}} = \boldsymbol{w} s_{\mathtt{rad}}$, where $\boldsymbol{w} \in \mathbb{C}^{N_{\mathtt{t}} \times 1}$ is the beamforming vector, $\mathbb{E}[s_{\mathtt{rad}}] = 0$ and $\mathbb{E}[\vert{}s_{\mathtt{rad}}\vert{}^2] = 1$.} Note that $s_k$, $s_{k'}$, and $s_{\mathtt{rad}}$ are assumed to be mutually independent~\cite{Salem_2025}.

The composite signal received at the AP is given by~\cite{Salem_2025}
\begin{align}\label{eq:y_AP_original}
\boldsymbol{y}_{\mathtt{AP}} = &\underbrace{\sum_{k \in \mathcal{K}_{\mathtt{MEC}} } \boldsymbol{h}_k \sqrt{p_k}\, s_k}_{\text{MEC Offloading Signals}} + \underbrace{\sum_{k' \in \mathcal{K}_{\mathtt{com}}  } \boldsymbol{h}_{k'} \sqrt{p_{k'}}\, s_{k'}}_{\text{Conventional Uplink Signals}} \nonumber\\ 
&+\underbrace{\alpha \boldsymbol{G}(\theta,\phi)\boldsymbol{w} s_{\mathtt{rad}}}_{\text{Target Echo}} + \underbrace{\boldsymbol{H}_{\mathtt{SI}}\boldsymbol{w} s_{\mathtt{rad}}}_{\text{Residual SI}} + \boldsymbol{n}_c
\end{align}
where $\alpha$ is the reflection coefficient. The two-way target response matrix is $\boldsymbol{G}(\theta, \phi) = \boldsymbol{a}_{\mathtt{r}}(\theta, \phi)\boldsymbol{a}_{\mathtt{t}}^{\mathtt{H}}(\theta, \phi) \in \mathbb{C}^{N_{\mathtt{r}} \times N_{\mathtt{t}}}$, with azimuth $\theta \in [-\pi, \pi]$ and elevation $\phi \in [-\pi/2, \pi/2]$ of the target relative to the AP. Assuming a UPA on the Y-Z plane with half-wavelength spacing, {the steering vector $\boldsymbol{a}_{\mathtt{b}}(\theta, \phi)$ for $\mathtt{b} \in \{t, r\}$} (denoting transmission and reception, respectively) is given as
\begin{equation}\label{eq:steering_vector_definition}
    \boldsymbol{a}_{\mathtt{b}}(\theta, \phi) = \boldsymbol{a}_{\mathtt{b},\mathtt{h}}(\theta, \phi) \otimes \boldsymbol{a}_{\mathtt{b},\mathtt{v}}(\phi) \in \mathbb{C}^{N_{\mathtt{b}} \times 1}
\end{equation}
where $\otimes$ is the Kronecker product, and $\boldsymbol{a}_{\mathtt{b},\mathtt{h}}(\theta, \phi) \in \mathbb{C}^{N_{\mathtt{b},\mathtt{h}} \times 1}$ and $\boldsymbol{a}_{\mathtt{b},\mathtt{v}}(\phi) \in \mathbb{C}^{N_{\mathtt{b},\mathtt{v}} \times 1}$ denote the centered horizontal and vertical steering vectors of the UPA, whose elements are compactly expressed as
\begin{align}\label{eq:resonponse_vector_h}
    [\boldsymbol{a}_{\mathtt{b},\mathtt{h}}(\theta, \phi)]_n &= e^{\left[ j \pi \left(n - \frac{N_{\mathtt{b},\mathtt{h}} + 1}{2}\right) \sin\theta \sin\phi \right]}\\
    \label{eq:response_vector_v}
    [\boldsymbol{a}_{\mathtt{b},\mathtt{v}}(\phi)]_q &= e^{\left[ j \pi \left(q - \frac{N_{\mathtt{b},\mathtt{v}} + 1}{2}\right) \cos\phi \right]}
\end{align}
for $n = \{1, \dots, N_{\mathtt{b},\mathtt{h}}\}$ and $q = \{1, \dots, N_{\mathtt{b},\mathtt{v}}\}$, where $N_{\mathtt{b},\mathtt{h}}$ and $N_{\mathtt{b},\mathtt{v}}$ denote the numbers of antennas along the horizontal and vertical dimensions, respectively, satisfying $N_{\mathtt{b},\mathtt{h}}N_{\mathtt{b},\mathtt{v}} = N_{\mathtt{b}}$. 

In Eq. \eqref{eq:y_AP_original}, the RSI channel matrix after active cancellation is modeled as $\boldsymbol{H}_{\mathtt{SI}} = \sqrt{\chi_{\mathtt{SI}}} \boldsymbol{G}_{\mathtt{SI}}^{\mathtt{H}} \in \mathbb{C}^{N_{\mathtt{r}} \times N_{\mathtt{t}}}$ \cite{JSAC_FD18}, where $\chi_{\mathtt{SI}}$ is the RSI power factor. The fading matrix $\boldsymbol{G}_{\mathtt{SI}} \in \mathbb{C}^{N_{\mathtt{t}} \times N_{\mathtt{r}}}$ follows a Rician channel model with a Rician factor $\kappa_{\mathtt{SI}}$, given by
\begin{equation}\label{eq:G_si}
\boldsymbol{G}_{\mathtt{SI}} = \sqrt{\frac{\kappa_{\mathtt{SI}}}{\kappa_{\mathtt{SI}}+1}} \boldsymbol{G}_{\mathtt{SI}}^{\mathtt{LoS}} + \sqrt{\frac{1}{\kappa_{\mathtt{SI}}+1}} \boldsymbol{G}_{\mathtt{SI}}^{\mathtt{NLoS}}
\end{equation}
where $\boldsymbol{G}_{\mathtt{SI}}^{\mathtt{LoS}}$ is the deterministic line-of-sight (LoS) component, and $\boldsymbol{G}_{\mathtt{SI}}^{\mathtt{NLoS}}$ is the non-LoS (NLoS) component whose elements are independent and identically distributed (i.i.d.) $\mathcal{CN}(0, 1)$. Finally, $\boldsymbol{n}_c \sim \mathcal{CN}(\boldsymbol{0}, \sigma_c^2 \boldsymbol{I}_{N_{\mathtt{r}}})$ denotes the additive white Gaussian noise (AWGN).

To decode MEC data and suppress multi-user interference (MUI), the AP applies a linear ZF receiver, $\boldsymbol{U}_{\mathtt{ZF}} = \hat{\boldsymbol{H}}(\hat{\boldsymbol{H}}^{\mathtt{H}} \hat{\boldsymbol{H}})^{-1}$, where {$\hat{\boldsymbol{H}} = [\hat{\boldsymbol{h}}_1, \cdots, \hat{\boldsymbol{h}}_k, \cdots, \hat{\boldsymbol{h}}_K] \in \mathbb{C}^{N_{\mathtt{r}} \times K}$ denotes the estimated channel matrix}. The post-processed signal for MEC user $k$ is obtained by applying the $k$-th column of $\boldsymbol{U}_{\mathtt{ZF}}$, denoted as $\boldsymbol{u}_k$, to the received signal in~\eqref{eq:y_AP_original}: 
\begin{align}\label{eq:y_k_signal}
\hat{y}_k &= \boldsymbol{u}_k^{\mathtt{H}} \boldsymbol{y}_{\mathtt{AP}} = \sqrt{p_k} s_k + \!\!\!\!\sum_{j \in \mathcal{K}_{\mathtt{MEC}}} \!\!\!\! \boldsymbol{u}_k^{\mathtt{H}} \boldsymbol{e}_j \sqrt{p_j} s_j + \!\!\!\!\!\sum_{k' \in \mathcal{K}_{\mathtt{com}}} \!\!\!\! \boldsymbol{u}_k^{\mathtt{H}} \boldsymbol{e}_{k'} \sqrt{p_{k'}} s_{k'}   \nonumber\\
& + \boldsymbol{u}_k^{\mathtt{H}} \big( \alpha \boldsymbol{G}(\theta, \phi) + \boldsymbol{H}_{\mathtt{SI}} \big) \boldsymbol{w} s_{\mathtt{rad}} + \boldsymbol{u}_k^{\mathtt{H}} \boldsymbol{n}_c.
\end{align}
\subsection{Radar Sensing Model}
The FD-AP aims to estimate the target's angles, $\psi \in \{\theta, \phi\}$. Following~\cite{Jinho_ISAC_2026,Hassanien_radar_2010}, we focus on the steady-state target tracking phase, where the AP uses a prior location estimate to form a directional radar beam $\boldsymbol{w}$. 

 Let $\boldsymbol{v} \triangleq \text{vec}(\boldsymbol{G}(\theta, \phi)\boldsymbol{X}) \in \mathbb{C}^{N_{\mathtt{r}} L \times 1}$ denote the vectorized radar sensing signal, where  $\boldsymbol{X} = \boldsymbol{w} \boldsymbol{s}_{rad}^{\mathtt{T}} \in \mathbb{C}^{N_{\mathtt{t}} \times L}$ represents the overall radar transmit signal matrix over the $L$ snapshots~\cite{Liu_ISAC_2018}. Moreover, $\boldsymbol{s}_{rad} = \big[ s_{rad}(1), s_{rad}(2), \dots, s_{rad}(L) \big]^{\mathtt{T}} \in \mathbb{C}^{L \times 1}$ denotes the radar pulse sequence.

By vectorizing $L$ snapshots of $\boldsymbol{y}_{\mathtt{AP}}$, the equivalent radar signal is given by~\cite{Xu_ISAC_2022}:
\begin{align}\label{eq:ys}
    \boldsymbol{y}_s = \alpha \boldsymbol{v} + \boldsymbol{n}_{\mathtt{eq}}
\end{align}
where $\boldsymbol{n}_{\mathtt{eq}} \in \mathbb{C}^{N_{\mathtt{r}} L \times 1}$ represents the vectorized equivalent interference-plus-noise vector. Specifically, $\boldsymbol{n}_{\mathtt{eq}}$ is formed by stacking the spatial snapshot vectors $\boldsymbol{n}_{\mathtt{eq}}(l) \in \mathbb{C}^{N_{\mathtt{r}} \times 1}$ for $l \in \{1, 2, \dots, L\}$, where the instantaneous interference-plus-noise at the $l$-th snapshot is given as:
\begin{align}
    \boldsymbol{n}_{\mathtt{eq}}(l) &= \sum_{k \in \mathcal{K}_{\mathtt{MEC}}} \boldsymbol{h}_k \sqrt{p_k} s_k(l) + \sum_{k' \in \mathcal{K}_{\mathtt{com}}} \boldsymbol{h}_{k'} \sqrt{p_{k'}} s_{k'}(l) \nonumber \\ 
    &\quad + \boldsymbol{n}_{\mathtt{SI}}(l) + \boldsymbol{n}_c(l)
\end{align}
where the instantaneous RSI and AWGN vectors are $\boldsymbol{n}_{\mathtt{SI}}(l) = \boldsymbol{H}_{\mathtt{SI}} \boldsymbol{w} s_{\mathtt{rad}}(l) = \sqrt{\chi_{\mathtt{SI}}} \boldsymbol{G}_{\mathtt{SI}}^{\mathtt{H}} \boldsymbol{w} s_{\mathtt{rad}}(l)$ and $\boldsymbol{n}_c(l) \sim \mathcal{CN}(\boldsymbol{0}, \sigma_c^2 \boldsymbol{I}_{N_{\mathtt{r}}})$.

\section{Communication and Sensing Performance}\label{sec_System_analysis}
This section derives uplink rate lower bounds under ZF reception, along with sensing CRLBs to evaluate $\theta$ and $\phi$ estimation accuracy.
\subsection{Communication Performance}
 From~\eqref{eq:y_k_signal}, the achievable uplink rate of user $k$ is given as
{\small\begin{align}\label{eq:Rate_MC}
    R_k^\mathtt{cp} = B \log_2 \Big( 1  + \frac{\mathbb{E} \left[ |\sqrt{p_k} s_k|^2 \right]}{\mathbb{E}[I_k^{\mathtt{ce}}] + \mathbb{E}[I_k^{\mathtt{ec}}] + \mathbb{E}[I_k^{\mathtt{rsi}}] + \mathbb{E}[|\boldsymbol{u}_k^{\mathtt{H}} \boldsymbol{n}_c|^2] } \Big)
\end{align}}where $\rho = \Vert{}\boldsymbol{w}\Vert{}^2$ and $I_k^{\mathtt{ce}} \triangleq \sum_{j \in \mathcal{K}_{\mathtt{MEC}}} p_j |\boldsymbol{u}_k^{\mathtt{H}} \boldsymbol{e}_j|^2 + \sum_{k' \in \mathcal{K}_{\mathtt{com}}} p_{k'} |\boldsymbol{u}_k^{\mathtt{H}} \boldsymbol{e}_{k'}|^2$ denotes the inter-user interference caused by CSI estimation errors. Furthermore, $I_k^{\mathtt{ec}} \triangleq \rho |\alpha|^2 |\boldsymbol{u}_k^{\mathtt{H}} \boldsymbol{G}(\theta, \phi) \tilde{\boldsymbol{w}}|^2$ represents the instantaneous target-echo interference power, while $I_k^{\mathtt{rsi}} \triangleq \rho |\boldsymbol{u}_k^{\mathtt{H}} \boldsymbol{H}_{\mathtt{SI}} \tilde{\boldsymbol{w}}|^2$ denotes the structured RSI power resulting from loopback coupling. Similarly, we can obtain the communication-user uplink rate $R_{k'}^{\mathtt{cm}}$ by replacing $k \in \mathcal{K}_{\mathtt{MEC}}$ with $k' \in \mathcal{K}_{\mathtt{com}}$ in \eqref{eq:Rate_MC}. 

The following theorem presents the closed-form expression for the achievable uplink rate of the MEC UEs and communication UEs.
\begin{theorem}\label{theorem_uplink_rate}
Using the use-and-then-forget (UatF) bounding technique and the inverse Wishart property $1/\mathbb{E}[\|\boldsymbol{u}_k\|^2] = (N_{\mathtt{r}} - K)\xi_k$, the achievable closed-form offloading rate for the priority MEC user $k$ is given by
\begin{equation}\label{eq:up_link_rate}
    R_k^\mathtt{cp}(\boldsymbol{p}, \rho) = B \log_2 \left( 1 + \frac{p_k (N_{\mathtt{r}} - K) \xi_k}{I^{\mathtt{er}}(\boldsymbol{p}) + I^\mathtt{se}(\rho) + \sigma_c^2} \right)
\end{equation}
where $\boldsymbol{p} \triangleq \{p_k, p_{k'}\}_{\forall k, k' \in \mathcal{K}}$, $I^{\mathtt{er}}(\boldsymbol{p}) \triangleq \sum_{j \in \mathcal{K}_{\mathtt{MEC}} } p_j \epsilon_j + \sum_{k' \in \mathcal{K}_{\mathtt{com}} } p_{k'} \epsilon_{k'}$ represents the inter-user interference, and the co-site radar cross-interference factor $I^\mathtt{se}(\rho)$ is given as
\begin{equation}\label{eq:I_k_sen_final}
    I^\mathtt{se}(\rho) \triangleq \rho \left( |\alpha|^2 N_{\mathtt{t}} + \frac{\chi_{\mathtt{SI}} \kappa_{\mathtt{SI}} \mathcal{D}_{\mathtt{H}}(\theta, \phi) \mathcal{D}_{\mathtt{V}}(\phi)}{N_{\mathtt{t}} (\kappa_{\mathtt{SI}}+1)} + \frac{\chi_{\mathtt{SI}}}{\kappa_{\mathtt{SI}}+1} \right).
\end{equation}
Here, $\mathcal{D}_{\mathtt{H}}(\theta, \phi)$ and $\mathcal{D}_{\mathtt{V}}(\phi)$ denote the expanded spatial Dirichlet beamforming kernels for the horizontal and vertical transmit sub-arrays, respectively, defined as:
\begin{align}
    \mathcal{D}_{\mathtt{H}}(\theta, \phi) &\triangleq \left[ \frac{\sin\left(\frac{N_{\mathtt{t},\mathtt{h}} \pi}{2} (\sin\theta \sin\phi - \sin\theta_{\mathtt{SI}} \sin\phi_{\mathtt{SI}})\right)}{\sin\left(\frac{\pi}{2} (\sin\theta \sin\phi - \sin\theta_{\mathtt{SI}} \sin\phi_{\mathtt{SI}})\right)} \right]^2 \\
    \mathcal{D}_{\mathtt{V}}(\phi) &\triangleq \left[ \frac{\sin\left(\frac{N_{\mathtt{t},\mathtt{v}} \pi}{2} (\cos\phi - \cos\phi_{\mathtt{SI}})\right)}{\sin\left(\frac{\pi}{2} (\cos\phi - \cos\phi_{\mathtt{SI}})\right)} \right]^2
\end{align}
where $\theta_{\mathtt{SI}}$ and $\phi_{\mathtt{SI}}$ are the internal physical leakage angles from the transmit to the receive array panel at the full-duplex access point.

Symmetrically, the closed-form achievable uplink communication rate for user $k' \in \mathcal{K}_{\mathtt{com}}$ is given as
\begin{equation}\label{eq:comm_rate}
    R_{k'}^\mathtt{cm}(\boldsymbol{p}, \rho) = B \log_2 \left( 1 + \frac{p_{k'} (N_{\mathtt{r}} - K) \xi_{k'}}{I^{\mathtt{er}}(\boldsymbol{p}) + I^\mathtt{se}(\rho) + \sigma_c^2} \right).
\end{equation}
\end{theorem}
\begin{proof}
    See Appendix \ref{proof_uplink_rate}.
\end{proof}

Eqs.~\eqref{eq:up_link_rate} and \eqref{eq:comm_rate} show that increasing radar power raises RSI and target-echo interference, reducing uplink rates.
\begin{remark}\label{remark_dof_limits}
For fixed $K$, $N_{\mathtt{t}}$, pilot powers, and radar power, the rate in \eqref{eq:up_link_rate} grows as $B\log_2 N_{\mathtt{r}}+\mathcal{O}(1)$ when user powers are fixed. If all user powers instead scale as $p_u=P_u^{\mathtt{ref}}/N_{\mathtt{r}}$, $u\in\mathcal K$, then \cite{Ngo2013}
\begin{equation}\label{eq:ultimate_upper_bound}
\bar R_k^{\mathtt{cp},\infty}(\rho)=B\log_2\left(1+\frac{P_k^{\mathtt{ref}}\xi_k}{I^{\mathtt{se}}(\rho)+\sigma_c^2}\right).
\end{equation}
Thus, radar interference remains in the limiting rate under this power-scaling regime.
\end{remark}
\begin{proof}
See Appendix~\ref{proof_remark_dof}.
\end{proof}
\subsection{Sensing Performance}
We first characterize the interference-plus-noise covariance, which is required to derive the angular CRLBs. Lemma \ref{LM_lemma1} gives this covariance under the adopted interference assumptions.

\textit{1) Multi-source interference:} To evaluate the ultimate sensing limits, Lemma~\ref{LM_lemma1} captures the spatial correlation of this interference via the equivalent covariance matrix of $\boldsymbol{n}_{\mathtt{eq}}$.
\begin{lemma}\label{LM_lemma1}
Due to the intractability of colored RSI in FD-mMIMO ISCC systems, we upper-bound it as an equivalent spatially white Gaussian process \cite{Salem_2025}. Over a large coherence processing interval ($L \gg 1$), the aggregate interference-plus-noise covariance matrix simplifies to a scaled identity:
\begin{equation}\label{eq:lemma1}
    \boldsymbol{R}_{\mathtt{in}} \triangleq \mathbb{E}\{\boldsymbol{n}_{\mathtt{eq}}\boldsymbol{n}_{\mathtt{eq}}^{\mathtt{H}}\} = \left( I^{\mathtt{up}}(\boldsymbol{p}) + \sigma_{\mathtt{rad}}^2(\rho) \right) \boldsymbol{I}_{N_{\mathtt{r}} L}
\end{equation}
where $I^{\mathtt{up}}(\boldsymbol{p}) \triangleq \sum_{k \in \mathcal{K}_{\mathtt{MEC}}} p_k \beta_k + \sum_{k' \in \mathcal{K}_{\mathtt{com}}} p_{k'} \beta_{k'}$ and $\sigma_{\mathtt{rad}}^2(\rho) \triangleq \chi_{\mathtt{SI}} \rho + \sigma_c^2$ denote the aggregate uplink interference and the effective radar noise variance, respectively.
\end{lemma}
\begin{proof}
See Appendix~\ref{app:proofoflemma1}.
\end{proof}
Eq. \eqref{eq:lemma1} reveals that the FD-mMIMO interference scales with the number of UEs $K$ and the RSI level $\chi_{\mathtt{SI}}$. Using the interference covariance in (17), we derive the CRLBs for the target angles $(\theta, \phi)$.

\textit{2) The CRLB analysis:}  Let $\boldsymbol{\omega} \triangleq [\theta, \phi, \bar{\boldsymbol{\alpha}}^{\mathtt{T}}]^{\mathtt{T}} \in \mathbb{R}^{4 \times 1}$ be the vector of unknown parameters to be estimated from \eqref{eq:ys}, where $\bar{\boldsymbol{\alpha}} = [\Re(\alpha), \Im(\alpha)]^{\mathtt{T}}$. The maximum likelihood estimate (MLE) of $\boldsymbol{\omega}$ is given by
\begin{equation}\label{eq:MLE_w}
    \boldsymbol{\omega}_{\mathtt{MLE}} = \arg\min_{\boldsymbol{\omega}} \|\boldsymbol{y}_s - \alpha \boldsymbol{v}\|^2.
\end{equation}
The Fisher information matrix (FIM) for estimating $\boldsymbol{\omega}$, denoted by $\boldsymbol{J}_{\boldsymbol{\omega}}$, is derived from the log-likelihood function and is governed by the inverse covariance $\boldsymbol{R}_{\mathtt{in}}^{-1}$ as \cite{Liu2021CRB}:
\begin{align}\label{eq:fisher_information_matrix}
   \boldsymbol{J}_{\boldsymbol{\omega}} = 2 \Re \left( \frac{\partial (\alpha\boldsymbol{v})^{\mathtt{H}}}{\partial \boldsymbol{\omega}} \boldsymbol{R}_{\mathtt{in}}^{-1} \frac{\partial (\alpha\boldsymbol{v})}{\partial \boldsymbol{\omega}} \right).
\end{align}
Substituting \eqref{eq:lemma1} into~\eqref{eq:fisher_information_matrix} and factoring out $L$, we isolate the power-dependent scaling factor $\Psi(\boldsymbol{p}, \rho)$ as follows:
\begin{equation}
    \Psi(\boldsymbol{p}, \rho) \triangleq \frac{2L}{I^{\mathtt{up}}(\boldsymbol{p}) + \sigma_{\mathtt{rad}}^2(\rho)}.
\end{equation}
 {Accordingly, we partition $\boldsymbol{J}_{\boldsymbol{\omega}}$ into blocks as \cite{Bekkerman2006MIMO}:}
\begin{align}\label{eq:FIM_elements}
    \boldsymbol{J}_{\boldsymbol{\omega}} = \begin{bmatrix} J_{\theta\theta} & J_{\theta\phi} & \boldsymbol{j}_{\theta\bar{\boldsymbol{\alpha}}} \\ J_{\theta\phi} & J_{\phi\phi} & \boldsymbol{j}_{\phi\bar{\boldsymbol{\alpha}}} \\ \boldsymbol{j}_{\theta\bar{\boldsymbol{\alpha}}}^{\mathtt{T}} & \boldsymbol{j}_{\phi\bar{\boldsymbol{\alpha}}}^{\mathtt{T}} & \boldsymbol{J}_{\bar{\boldsymbol{\alpha}}\bar{\boldsymbol{\alpha}}} \end{bmatrix}. 
\end{align}
Letting $\tilde{\boldsymbol{w}} \triangleq \boldsymbol{w}/\sqrt{\rho}$ denote the normalized sensing direction, the rank-one transmit covariance is $\boldsymbol{w} \boldsymbol{w}^{\mathtt{H}} = \rho \tilde{\boldsymbol{w}} \tilde{\boldsymbol{w}}^{\mathtt{H}}$. Aggregating over $L$ snapshots, each FIM sub-block factorizes as $J_{xy} = \Psi(\boldsymbol{p}, \rho) \rho \tilde{c}_{xy}(\tilde{\boldsymbol{w}})$, where 
\begin{align}
    &\tilde{c}_{\theta\theta}(\tilde{\boldsymbol{w}}) \triangleq |\alpha|^2 \text{Tr}(\dot{\boldsymbol{G}}_\theta \tilde{\boldsymbol{w}} \tilde{\boldsymbol{w}}^{\mathtt{H}} \dot{\boldsymbol{G}}_\theta^{\mathtt{H}}) \label{eq:C_theta_theta}  \\
    &\tilde{c}_{\phi\phi}(\tilde{\boldsymbol{w}}) \triangleq |\alpha|^2 \text{Tr}(\dot{\boldsymbol{G}}_\phi \tilde{\boldsymbol{w}} \tilde{\boldsymbol{w}}^{\mathtt{H}} \dot{\boldsymbol{G}}_\phi^{\mathtt{H}})\label{eq:c_phi_phi}  \\
    &\tilde{c}_{\theta\phi}(\tilde{\boldsymbol{w}}) \triangleq \Re\left( |\alpha|^2 \text{Tr}(\dot{\boldsymbol{G}}_\theta \tilde{\boldsymbol{w}} \tilde{\boldsymbol{w}}^{\mathtt{H}} \dot{\boldsymbol{G}}_\phi^{\mathtt{H}}) \right)\label{eq:c_theta_phi} \\
    &\boldsymbol{\tilde{c}}_{\bar{\boldsymbol{\alpha}}\bar{\boldsymbol{\alpha}}}(\tilde{\boldsymbol{w}}) \triangleq \text{Tr}(\boldsymbol{G} \tilde{\boldsymbol{w}} \tilde{\boldsymbol{w}}^{\mathtt{H}} \boldsymbol{G}^{\mathtt{H}}) \boldsymbol{I}_2\label{eq:c_alpha_alpha}  \\
    &\boldsymbol{\tilde{c}}_{\theta\bar{\boldsymbol{\alpha}}}(\tilde{\boldsymbol{w}}) \triangleq \Re\left( \alpha^* \text{Tr}(\boldsymbol{G} \tilde{\boldsymbol{w}} \tilde{\boldsymbol{w}}^{\mathtt{H}} \dot{\boldsymbol{G}}_\theta^{\mathtt{H}})[1, j] \right)\label{eq:c_theta_alpha}  \\
    &\boldsymbol{\tilde{c}}_{\phi\bar{\boldsymbol{\alpha}}}(\tilde{\boldsymbol{w}}) \triangleq \Re\left( \alpha^* \text{Tr}(\boldsymbol{G} \tilde{\boldsymbol{w}} \tilde{\boldsymbol{w}}^{\mathtt{H}} \dot{\boldsymbol{G}}_\phi^{\mathtt{H}})[1, j] \right)\label{eq:c_phi_alpha} 
\end{align}
and $\dot{\boldsymbol{G}}_\theta = \dot{\boldsymbol{a}}_{\mathtt{r},\theta} \boldsymbol{a}_{\mathtt{t}}^{\mathtt{H}} + \boldsymbol{a}_{\mathtt{r}} \dot{\boldsymbol{a}}_{\mathtt{t},\theta}^{\mathtt{H}}$ and $\dot{\boldsymbol{G}}_\phi = \dot{\boldsymbol{a}}_{\mathtt{r},\phi} \boldsymbol{a}_{\mathtt{t}}^{\mathtt{H}} + \boldsymbol{a}_{\mathtt{r}} \dot{\boldsymbol{a}}_{\mathtt{t},\phi}^{\mathtt{H}}$. 

Applying the Schur complement to \eqref{eq:FIM_elements}, the equivalent angular FIM $\boldsymbol{J}_{\theta,\phi} \in \mathbb{R}^{2\times 2}$ is \cite{Huang_ISAC_2022}:
\begin{align}\label{eq:equip_FIM}
    \boldsymbol{J}_{\theta,\phi} = \begin{bmatrix} J_{\theta\theta} & J_{\theta\phi} \\ J_{\theta\phi} & J_{\phi\phi} - \boldsymbol{j}_{\phi\bar{\boldsymbol{\alpha}}} \boldsymbol{J}_{\bar{\boldsymbol{\alpha}}\bar{\boldsymbol{\alpha}}}^{-1} \boldsymbol{j}_{\phi\bar{\boldsymbol{\alpha}}}^{\mathtt{T}} \end{bmatrix}.
\end{align}
The diagonal entries of the inverse angular FIM $\boldsymbol{J}_{\theta,\phi}^{-1}$ give the CRLBs for unbiased angle estimators. Theorem 2 presents the resulting expressions.
\begin{theorem}\label{theorem_CRLB}
   The CRLBs for $\theta$ and $\phi$ in the FD-mMIMO ISCC system admit the following closed-form expressions:
    \begin{align}
        \text{CRLB}_\theta(\boldsymbol{p}, \rho) &= \frac{I^{\mathtt{up}}(\boldsymbol{p}) + \sigma_{\mathtt{rad}}^2(\rho)}{2L \rho} \mathcal{G}_\theta(\tilde{\boldsymbol{w}}) \label{eq:CRLB_theta} \\
        \text{CRLB}_\phi(\boldsymbol{p}, \rho) &= \frac{I^{\mathtt{up}}(\boldsymbol{p}) + \sigma_{\mathtt{rad}}^2(\rho)}{2L \rho} \mathcal{G}_\phi(\tilde{\boldsymbol{w}}) \label{eq:CRLB_phi}
    \end{align}
    where $\mathcal{G}_\theta(\tilde{\boldsymbol{w}})$ and $\mathcal{G}_\phi(\tilde{\boldsymbol{w}})$ are power-independent geometric multipliers defined by
    \begin{align} 
        \mathcal{G}_\theta(\tilde{\boldsymbol{w}}) &\triangleq \left( \tilde{c}_{\theta\theta}(\tilde{\boldsymbol{w}}) - \frac{|\tilde{c}_{\theta\phi}(\tilde{\boldsymbol{w}})|^2}{\tilde{c}_{\phi\phi}(\tilde{\boldsymbol{w}}) - \frac{\|\tilde{\boldsymbol{c}}_{\phi\bar{\boldsymbol{\alpha}}}(\tilde{\boldsymbol{w}})\|^2}{\tilde{c}_{\bar{\boldsymbol{\alpha}}\bar{\boldsymbol{\alpha}}}(\tilde{\boldsymbol{w}})}} \right)^{-1} \label{eq:steering_vector_theta} \\
        \mathcal{G}_\phi(\tilde{\boldsymbol{w}}) &\triangleq \left( \tilde{c}_{\phi\phi}(\tilde{\boldsymbol{w}}) - \frac{|\tilde{c}_{\theta\phi}(\tilde{\boldsymbol{w}})|^2}{\tilde{c}_{\theta\theta}(\tilde{\boldsymbol{w}}) - \frac{\|\tilde{\boldsymbol{c}}_{\theta\bar{\boldsymbol{\alpha}}}(\tilde{\boldsymbol{w}})\|^2}{\tilde{c}_{\bar{\boldsymbol{\alpha}}\bar{\boldsymbol{\alpha}}}(\tilde{\boldsymbol{w}})}} \right)^{-1}. \label{eq:steering_vector_phi}
    \end{align}
\end{theorem}
\begin{proof}
    To derive the closed-form CRLBs, we invert the decoupled $2 \times 2$ angular FIM $\boldsymbol{J}_{\theta,\phi}$ in \eqref{eq:equip_FIM}. For a $2 \times 2$ matrix block, the diagonal entries of its inverse $\boldsymbol{J}_{\theta,\phi}^{-1}$ explicitly yield the effective parameter bounds. We use $J_{xy} = \Psi(\boldsymbol{p}, \rho) \rho \tilde{c}_{xy}(\tilde{\boldsymbol{w}})$ to isolate  $\mathcal{G}_\theta(\tilde{\boldsymbol{w}})$ and $\mathcal{G}_\phi(\tilde{\boldsymbol{w}})$. Factoring out $[\Psi(\boldsymbol{p}, \rho)\rho]^{-1}$ yields \eqref{eq:CRLB_theta} and \eqref{eq:CRLB_phi}.
\end{proof}

\begin{figure}[!htb]
	\centering
    \includegraphics[width=0.95\columnwidth,trim={0cm 0.0cm 0cm 0.0cm}]{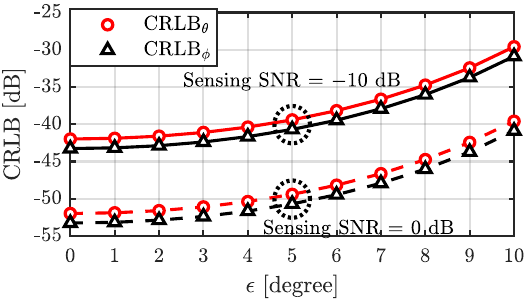}
    \vspace{-2mm}
	\caption{CRLB for $\theta$ and $\phi$ when the sensing beamformer $\boldsymbol{w}$ is steered towards $\tilde{\theta} = \theta + \epsilon$ and $\tilde{\phi} = \phi + \epsilon$, with $N_{\mathtt{t}} = 64$, $N_{\mathtt{r}} = 64$, $K_{\mathtt{MEC}} = 8$, $K_{\mathtt{com}} = 8$, and $L = 256$.}
	\label{fig:CRLB_vs_error}
\end{figure}
As seen in \eqref{eq:CRLB_theta} and \eqref{eq:CRLB_phi}, the CRLBs depend on the sensing beamformer $\boldsymbol{w}$. Fig.~\ref{fig:CRLB_vs_error} shows that the angular CRLBs increase with pointing error $\epsilon \in [0^\circ, 10^\circ]$. For tractable resource allocation, we next consider perfect beam alignment and obtain the simplified expressions of \eqref{eq:equip_FIM} in Remark~\ref{remark_tractability_CRLB}.
\begin{remark}\label{remark_tractability_CRLB}
For the case of perfect beam alignment, \textit{i.e.} $(\tilde{\theta}, \tilde{\phi}) = (\theta, \phi)$, the normalized transmit beamformer evaluates to $\tilde{\boldsymbol{w}} = \frac{1}{\sqrt{N_{\mathtt{t}}}}\boldsymbol{a}_{\mathtt{t}}$. Applying the spatial orthogonality properties of centered planar arrays ($\boldsymbol{a}_{\mathtt{t}}^{\mathtt{H}} \dot{\boldsymbol{a}}_{\mathtt{t},\theta} = 0$ and $\boldsymbol{a}_{\mathtt{t}}^{\mathtt{H}} \dot{\boldsymbol{a}}_{\mathtt{t},\phi} = 0$), the geometric trace constants $\tilde{c}_{(\cdot)}(\tilde{\boldsymbol{w}})$ explicitly simplify to:
\begin{align}
    &\tilde{c}_{\theta\theta} = |\alpha|^2 N_{\mathtt{t}} U_{\mathtt{r},\theta}, \quad \tilde{c}_{\phi\phi} = |\alpha|^2 N_{\mathtt{t}} U_{\mathtt{r},\phi},\label{eq:c_theta_theta_final} \\
    &\Re\{\tilde{c}_{\theta\phi}\} = |\alpha|^2 N_{\mathtt{t}} U_{\mathtt{r},\theta\phi},\label{eq:c_theta_phi_final} \\
    &\boldsymbol{\tilde{c}}_{\bar{\boldsymbol{\alpha}}\bar{\boldsymbol{\alpha}}} = N_{\mathtt{t}} N_{\mathtt{r}} \boldsymbol{I}_2, \quad \tilde{\boldsymbol{c}}_{\theta\bar{\boldsymbol{\alpha}}} = \boldsymbol{0}, \quad \tilde{\boldsymbol{c}}_{\phi\bar{\boldsymbol{\alpha}}} = \boldsymbol{0}.\label{eq:c_alpha_alpha_final}
\end{align}
Since $\tilde{\boldsymbol{c}}_{\theta\bar{\boldsymbol{\alpha}}} = \tilde{\boldsymbol{c}}_{\phi\bar{\boldsymbol{\alpha}}} = \boldsymbol{0}$, angular estimation decouples from $\bar{\boldsymbol{\alpha}}$, yielding the simplified CRLBs in \eqref{eq:CRLB_final_theta}--\eqref{eq:CRLB_final_phi}. For the $N_{\mathtt{rh}} \times N_{\mathtt{rv}}$ Y-Z planar receive array, the spatial energy terms are:
\begin{align}
    &U_{\mathtt{r},\theta} = \frac{N_{\mathtt{r}}}{12} \pi^2 \cos^2(\theta) \sin^2(\phi) (N_{\mathtt{rh}}^2-1), \\
    &U_{\mathtt{r},\phi} = \frac{N_{\mathtt{r}}}{12} \pi^2 \Big( (N_{\mathtt{rh}}^2-1) \sin^2(\theta) \cos^2(\phi) \nonumber\\ 
    &+ (N_{\mathtt{rv}}^2-1) \sin^2(\phi) \Big), \\
    &U_{\mathtt{r},\theta\phi} = \frac{N_{\mathtt{r}}}{12} \pi^2 \sin(\phi) \cos(\phi) \sin(\theta) \cos(\theta) (N_{\mathtt{rh}}^2-1).
\end{align}
In \eqref{eq:CRLB_final_theta}--\eqref{eq:CRLB_final_phi}, $N_{\mathtt{t}}$ scales the gain while $N_{\mathtt{r}}$ governs the angular terms. Furthermore, the CRLBs' proportionality to $I^{\mathtt{up}}(\boldsymbol{p})$ and inverse proportionality to $\rho$ dictate a fundamental trade-off, making joint power allocation essential to preserve sensing performance.
\end{remark}
\begin{proof}
    See Appendix~\ref{proof_CRLB}.
\end{proof}
\begin{figure*}[!b]
\normalsize
\hrulefill 
\begin{align} 
   \text{CRLB}_\theta(\boldsymbol{p}, \rho) &= \frac{6 \big( I^{\mathtt{up}}(\boldsymbol{p}) + \sigma_{\mathtt{rad}}^2(\rho) \big)}{L\rho |\alpha|^2 N_{\mathtt{t}} N_{\mathtt{r}} \pi^2} \left[ \frac{(N_{\mathtt{rh}}^2-1) \sin^2(\theta) \cos^2(\phi) + (N_{\mathtt{rv}}^2-1) \sin^2(\phi)}{\cos^2(\theta) \sin^4(\phi) (N_{\mathtt{rh}}^2-1)(N_{\mathtt{rv}}^2-1)} \right]  \label{eq:CRLB_final_theta}\\  
    \text{CRLB}_\phi(\boldsymbol{p}, \rho) &= \frac{6 \big( I^{\mathtt{up}}(\boldsymbol{p}) + \sigma_{\mathtt{rad}}^2(\rho) \big)}{L\rho |\alpha|^2 N_{\mathtt{t}} N_{\mathtt{r}} \pi^2 \sin^2(\phi) (N_{\mathtt{rv}}^2-1)}. \label{eq:CRLB_final_phi}
\end{align}
\end{figure*}
\begin{remark}\label{remark_Nt_CRLB}
    For fixed powers and target angles away from singular geometries, the angular CRLBs scale as $\mathcal{O}(1/N_{\mathtt{t}})$ with fixed $N_{\mathtt{r}}$. For a square receive UPA and fixed $N_{\mathtt{t}}$, they scale as $\mathcal{O}(1/N_{\mathtt{r}}^2)$.
\begin{proof}
    See Appendix~\ref{Proof_remark_Nt_CRLB}.
\end{proof}
\end{remark}
At finite array sizes, we can see that radar power improves sensing but reduces uplink rates, whereas user power increases sensing interference. These competing effects motivate the joint allocation, detailed next.

\section{ISCC Resource Allocation}\label{sec_Problem_Proposed_solution}
In this section, we formulate the resource allocation problem and propose an effective SCA-based solution for FD-mMIMO ISCC systems.

\subsection{Problem Formulation}\label{subsec_Problem_main} 
\noindent \textit{1) Computational model:} 
The uplink transmission time required for UE $k \in \mathcal{K}_{\mathtt{MEC}}$ to offload $\varphi_k D_k$ bits is given by
\begin{align}\label{t_up_origin}
    t_k^{\mathtt{up}} = \frac{\varphi_k D_k}{R_k^\mathtt{cp}(\boldsymbol{p}, \rho)}.
\end{align}
The edge computing time to execute the task with the allocated processing rate $f_k$ is given by
\begin{align}
    t_k^{\mathtt{es}} = \frac{\varphi_k C_k}{f_k}.
\end{align}
Concurrently, the local processing time at UE $k$, operating at frequency $f_k^{\mathtt{ue}}$, can be expressed as
\begin{align}
    t_k^{\mathtt{loc}} = \frac{(1-\varphi_k)C_k}{f_k^{\mathtt{ue}}}.
\end{align}
Since the local execution and the network-assisted offloading occur in parallel, the e2e latency for user $k$ is given by
\begin{align}\label{eq:total_latency}
    T_{k}^{\mathtt{e2e}}(\varphi_k, \boldsymbol{p}, \rho, \boldsymbol{f}) =  \max \big(t_{k}^{\mathtt{loc}} , t_p + t_{k}^{\mathtt{up}} + t_{k}^{\mathtt{es}}\big) .
\end{align}
Notably, the downlink latency for MEC tasks is negligible because computation results are small and the AP transmits at high power \cite{Son_2026}. Under asynchronous task execution, e2e latency is the maximum of local execution and the combined offloading and edge execution delay.

Let {$E_{k}^{\mathtt{loc}}$, $E_{k}^{\mathtt{up}}$ and $E_k^{\mathtt{pil}}$} denote the energy consumption for local computation, offloading transmission and pilot training of the priority MEC user $k$, respectively. The total energy consumption for user $k$ is given by~\cite{ling2021}
\begin{align}
 &E_{k} = E_{k}^{\mathtt{loc}} + E_{k}^{\mathtt{up}} + E_k^{\mathtt{pil}} \nonumber\\ 
 &= \frac{{{\mu _k}}}{2}(1-{\varphi _{k}}){ C_{k}}{(f_k^\mathtt{ue})^2} + {p_k}\frac{{ {\varphi_{k}} {D_{k}}}}{{R_k^{\mathtt{cp}}(\boldsymbol{p}, \rho)} } + \vartheta_u \tau_p T_s
\end{align}
where $\mu_k/2$ (in $\text{W}\cdot\text{s}^3/\text{cycle}^3$) denotes the computational energy parameter of UE $k$. Note that the communication users $k' \in \mathcal{K}_{\mathtt{com}}$ do not execute any computational tasks. Consequently, their energy consumption arises only from small data transmission and pilot training, subject to their maximum power constraints.

\noindent\textit{2) Resource allocation problem:} We jointly optimize offloading fractions $\boldsymbol\varphi$, user powers $\boldsymbol p$, radar power $\rho$, and edge CPU allocations $\boldsymbol f$ to minimize average e2e latency:
\begin{subequations} \label{eq:Main}
	\begin{IEEEeqnarray}{cl}
		 \mathop {\min}\limits_{\boldsymbol{\varphi}, \boldsymbol{p}, \rho, \boldsymbol{f}} &\quad \frac{1}{K_{\mathtt{MEC}} }\sum_{k \in \mathcal{K}_{\mathtt{MEC}} } T_k^{\mathtt{e2e}}  \label{eq:Maina} \\ 
		\st& T_{k}^{\mathtt{e2e}} \le T_k^{\max},  \forall k \in \mathcal{K}_{\mathtt{MEC}}   \label{eq:Mainb}\\
		& E_{k} \le E_k^{\max },  \forall k \in \mathcal{K}_{\mathtt{MEC}}  \label{eq:Mainc} \\
        & R_{k'}^{\mathtt{cm}}(\boldsymbol{p}, \rho) \ge R^{\min},  \forall k' \in \mathcal{K}_{\mathtt{com}} \label{eq:Maind} \\
        & \text{CRLB}_\psi(\boldsymbol{p} , \rho) \le \Phi_\psi^0, \psi \in \{\theta, \phi\} \label{eq:Maine} \\
        & \sum_{k \in \mathcal{K}_{\mathtt{MEC}}} f_k \le F^{\max}_{es} \label{eq:Mainf} \\
         & 0 \le \rho \le P^{\max}_{\mathtt{AP}}, 0 \le p_k \le P_k^{\max},   \forall k \in \mathcal{K} \label{eq:Maing} \\
         &  0 \le \varphi_k \le 1,   \forall k \in \mathcal{K}_{\mathtt{MEC}}. \label{eq:Mainh}
	\end{IEEEeqnarray}
\end{subequations}
Constraints~\eqref{eq:Mainb}--\eqref{eq:Maine} impose task deadlines, energy budgets, minimum communication rates, and sensing accuracy. Constraints~\eqref{eq:Mainf}--\eqref{eq:Mainh} specify CPU, power, and offloading limits. The coupled latency and energy terms make problem~\eqref{eq:Main} nonconvex.

\subsection{Proposed Solution}\label{subsection_proposed_solution}
 To reformulate the rate-dependent latency terms in \eqref{eq:Main}, we introduce the slack variables
  $\boldsymbol{\varrho} \triangleq \{ \varrho_k\}_{\forall k \in \mathcal{K}}$, satisfying
\begin{align}\label{eq:r_k}
   R_k^{\mathtt{cp}}(\boldsymbol{p}, \rho)  \ge \frac{1}{\varrho_k}, \forall k \in \mathcal{K}_{\mathtt{MEC}}. 
\end{align}
From \eqref{eq:r_k}, the objective function \eqref{eq:Maina} and constraint \eqref{eq:Mainb} can be rewritten as
\begin{IEEEeqnarray}{cl}
    T_k^{\mathtt{e2e}} \le \max\Big(t_k^{\mathtt{loc}},t_p + \varphi_k \varrho_k D_k + \frac{\varphi_k C_k}{f_k}\Big) := T_k^{'\mathtt{e2e}}.\quad
\end{IEEEeqnarray}
Similarly, constraint \eqref{eq:Mainc} can be rewritten as
\begin{align}\label{eq:UE_energy_equip}
     \frac{{{\mu _k}}}{2}(1-{\varphi _{k}}){ C_{k}}{(f_k^\mathtt{ue})^2} + {p_k}\varphi_{k}\varrho_k D_{k} + \vartheta_u \tau_p T_s \le E_k^{\max}.
\end{align}
As a result, problem \eqref{eq:Main} can be reformulated as
\begin{subequations} \label{eq:Main_relax}
	\begin{IEEEeqnarray}{cl}
		 \mathop {\min}\limits_{\boldsymbol{\varphi}, \boldsymbol{p}, \rho, \boldsymbol{f}, \boldsymbol{\varrho}} &\quad \frac{1}{K_{\mathtt{MEC}} }\sum_{k \in \mathcal{K}_{\mathtt{MEC}} } T_k^{'\mathtt{e2e}}  \label{eq:Main_relaxa} \\ 
		\st& T_{k}^{'\mathtt{e2e}} \le T_k^{\max},  \forall k \in \mathcal{K}_{\mathtt{MEC}}   \label{eq:Main_relaxb}\\
		&  \eqref{eq:Maind}, \eqref{eq:Maine}, \eqref{eq:Mainf}, \eqref{eq:Maing}, \eqref{eq:Mainh}, \eqref{eq:r_k},\eqref{eq:UE_energy_equip}.\quad\label{eq:Main_relaxc}. 
	\end{IEEEeqnarray}
\end{subequations}
We apply the SCA principles to the coupled rate, latency, and energy terms and rewrite the sensing constraint in affine form.

\textit{1) Convexifying constraint~\eqref{eq:r_k}:} We first convert $R_k^{\mathtt{cp}}(\boldsymbol{p}, \rho)$ into the form of $\ln{(1+x/y)}$ which we can approximate around the feasible points $x^{(i)} > 0$ and $y^{(i)} >0$ for $x$ and $y$ at the $i$-th iteration as follows~\cite{Nasir_2021}: 
\begin{align}\label{eq:R_K_SCA_Form}
    \ln \left( 1 + \frac{x}{y} \right) \geq & \ln \left( 1 + \frac{x^{(i)}}{y^{(i)}} \right) + 2 \frac{x^{(i)}}{x^{(i)} + y^{(i)}} \nonumber\\ 
& - \frac{(x^{(i)})^2}{x^{(i)} + y^{(i)}} \frac{1}{x} - \frac{x^{(i)}}{(x^{(i)} + y^{(i)})y^{(i)}} y.
\end{align}
Define $x_k=(N_{\mathtt{r}}-K)\xi_k p_k$ and $y=I^{\mathtt{er}}(\boldsymbol p)+I^{\mathtt{se}}(\rho)+\sigma_c^2$, with iteration values $x_k^{(i)}$ and $y^{(i)}$. Applying \eqref{eq:R_K_SCA_Form} gives the concave lower bound
\begin{align}\label{eq:R_K_SCA_final}
R_k^{\mathtt{cp},(i)}(\boldsymbol p,\rho)
=\frac{B}{\ln 2}\Bigg[&\ln\left(1+\frac{x_k^{(i)}}{y^{(i)}}\right)
+\frac{2x_k^{(i)}}{x_k^{(i)}+y^{(i)}}\nonumber\\
&-\frac{(x_k^{(i)})^2/x_k+x_k^{(i)}y/y^{(i)}}{x_k^{(i)}+y^{(i)}}\Bigg].
\end{align}
From \eqref{eq:R_K_SCA_final}, constraint \eqref{eq:r_k} is iteratively replaced by the following convex constraint:
\begin{align}\label{eq:r_k_equi}
   R_k^{\mathtt{cp}, (i)}(\boldsymbol{p}, \rho)  \ge \frac{1}{\varrho_k}, \forall k \in \mathcal{K}_{\mathtt{MEC}}. 
\end{align}
\textit{2) Convexifying constraint \eqref{eq:Main_relaxb}:} In this constraint, $t_k^{\mathtt{up}} = \varphi_k \varrho_k D_k$ and $t_k^{\mathtt{es}} = \varphi_k C_k/ f_k$ are still nonconvex due to the coupling among the variables ($\varphi_k$, $\varrho_k$, $f_k$). To convexify $t_k^{\mathtt{up}}$, we adopt the following inequality:
\begin{align}\label{eq:SCA_xy}
    xy \le \frac{1}{2}\Big( \frac{\overline{x}}{\overline{y}}y^2 + \frac{\overline{y}}{\overline{x}}x^2 \Big)
\end{align}
where $\overline{x}$ and $\overline{y}$ are feasible points of $x$ and $y$, respectively. Letting {$x= \varphi_k >0$}, $y=\varrho_k >0$, $\overline{x} = \varphi_k^{(i)} > 0 $, $\overline{y} = \varrho_k^{(i)} >0$, $t_k^{\mathtt{up}}$ is iteratively replaced by 
\begin{align}
    t_k^{\mathtt{up}} \le  \frac{D_k}{2} \Big(\frac{\varphi_k^{(i)}}{\varrho_k^{(i)}}\varrho_k^2 + \frac{\varrho_k^{(i)}}{ \varphi_k^{(i)}}\varphi_k^2 \Big) := t_k^{\mathtt{up}, (i)}.
\end{align}
Next, we convexify $t_k^{\mathtt{es}}= \varphi_k C_k / f_k$ by adopting the following inequality:
\begin{align}\label{eq:convexify_x_divide_y}
    \frac{x}{y} \le \frac{1}{2}\Big(\frac{x^2}{\overline{x} \,\overline{y} } + \frac{\overline{x} \,\overline{y}}{2y\overline{y} - \overline{y}^2} \Big)
\end{align}
where $\overline{x}>0$ and $\overline{y} >0$ denote feasible points at the previous iteration, and $y>\overline{y}/2$. By substituting $x=\varphi_k >0$, $\overline{x} = \varphi_k^{(i)} >0$, $y = f_k > 0$ and $\overline{y} = f_k^{(i)}>0$ in \eqref{eq:convexify_x_divide_y}, a convex upper bound on $t_k^{\mathtt{es}}$ is given by
\begin{align}
    t_k^{\mathtt{es}} \le& \frac{C_k}{2}\Big( \frac{\varphi_k^2}{\varphi_k^{(i)}f_k^{(i)}} + \frac{\varphi_k^{(i)}f_k^{(i)}}{2f_kf_k^{(i)} - (f_k^{(i)})^2}\Big):= t_k^{\mathtt{es},(i)}.
\end{align}
As a result, we have
\begin{align}
    T_k^{'\mathtt{e2e}} \le \max(t_k^{\mathtt{loc}}, t_p + t_k^{\mathtt{up},(i)} + t_k^{\mathtt{es},(i)}) := T_k^{'\mathtt{e2e},(i)}.
\end{align}
The function $T_k^{'\mathtt{e2e},(i)}$ is convex, giving the following convex approximation of constraint \eqref{eq:Main_relaxb}:
\begin{align}\label{eq:Te2e_SCA}
    T_k^{'\mathtt{e2e},(i)} \le T_k^{\max}, \forall k \in \mathcal{K}_{\mathtt{MEC}}.
\end{align}
\textit{3) Energy constraint \eqref{eq:UE_energy_equip}:} Constraint \eqref{eq:UE_energy_equip} is nonconvex due to the strong coupling between $p_k$, $\varphi_k$ and $\varrho_k$ variables. To convexify it, we adopt the following inequality:
\begin{align}
     xyz \le \frac{\overline{x}\, \overline{y} \, \overline{z}}{3}\Big( \frac{x^3}{\overline{x}^3} + \frac{y^3}{\overline{y}^3} + \frac{z^3}{\overline{z}^3}  \Big)
\end{align}
where $\overline{x} >0$, $\overline{y} >0$ and $\overline{z} >0$ are the feasible points of $x, y$ and $z$, respectively. Letting $x = p_k >0$, $\overline{x} = p_k^{(i)} >0$, $y = \varphi_k > 0$, $\overline{y} = \varphi_k^{(i)} >0$, $z = \varrho_k >0$ and $\overline{z} = \varrho_k^{(i)}$, constraint \eqref{eq:UE_energy_equip} is iteratively replaced by
\begin{align}\label{eq:SCA_energy_Const}
    &\frac{{{\mu _k}}}{2}(1-{\varphi _{k}}){ C_{k}}{(f_k^\mathtt{ue})^2} + \frac{D_kp_k^{(i)} \varphi_k^{(i)} \varrho_k^{(i)} }{3} \Big( \frac{p_k^3}{(p_k^{(i)})^3}  \nonumber\\ 
    & +\frac{\varphi_k^3}{(\varphi_k^{(i)})^3} + \frac{\varrho_k^3}{(\varrho_k^{(i)})^3} \Big) + \vartheta_u \tau_p T_s \le {E}_k^{\max}.
\end{align}
\textit{4) Communication-rate constraint \eqref{eq:Maind}:} Similarly, we can apply \eqref{eq:R_K_SCA_final} to convexify constraint \eqref{eq:Maind} by replacing $p_k, \forall k\in \mathcal{K}_{\mathtt{MEC}}$ with $p_{k'}, \forall k' \in \mathcal{K}_{\mathtt{com}}$. Thus, we can obtain
\begin{align}
    R_{k'}^{\mathtt{cm}}(\boldsymbol{p}, \rho) \ge R_{k'}^{\mathtt{cm}{(i)}}(\boldsymbol{p}, \rho).
\end{align}
As a result, constraint \eqref{eq:Maind} is iteratively replaced by
\begin{align}\label{eq:SCA_R_cm}
    R_{k'}^{\mathtt{cm}{(i)}}(\boldsymbol{p}, \rho) \ge R^{\min}, \forall k' \in \mathcal{K}_{\mathtt{com}}.
\end{align}
\textit{5) Sensing constraint:} For $\rho>0$, substituting \eqref{eq:CRLB_final_theta}--\eqref{eq:CRLB_final_phi} into \eqref{eq:Maine} gives the affine constraints
\begin{align}\label{eq:SCA_CRLB}
   & \frac{\Delta_{\psi}}{\Phi_{\psi}^0} \Big (\sum_{k \in \mathcal{K}_{\mathtt{MEC}}} p_k \beta_k + \sum_{k' \in \mathcal{K}_{\mathtt{com}}} p_{k'} \beta_{k'} \Big) \nonumber\\ 
    &+ \rho \Big(\frac{\Delta_{\psi}}{\Phi_{\psi}^0} \chi_{\mathtt{SI}} -1\Big) + \frac{\Delta_{\psi}}{\Phi_{\psi}^0}\sigma_c^2 \le 0, \psi \in \{\theta, \phi\}
\end{align}
where 
\begin{align}
    &\Delta_{\theta} = \frac{6 }{L |\alpha|^2 N_{\mathtt{t}} N_{\mathtt{r}} \pi^2} \\ \nonumber 
    &\,\times \left[ \frac{(N_{\mathtt{rh}}^2-1) \sin^2(\theta) \cos^2(\phi) + (N_{\mathtt{rv}}^2-1) \sin^2(\phi)}{\cos^2(\theta) \sin^4(\phi) (N_{\mathtt{rh}}^2-1)(N_{\mathtt{rv}}^2-1)} \right]\\ 
     &\Delta_{\phi} = \frac{6 }{L |\alpha|^2 N_{\mathtt{t}} N_{\mathtt{r}} \pi^2 \sin^2(\phi) (N_{\mathtt{rv}}^2-1)}.
\end{align}
Based on the above results, the convex program of \eqref{eq:Main} solved at iteration $i$ is given as
\begin{subequations} \label{eq:Main_convex}
	\begin{IEEEeqnarray}{cl}
		 \mathop {\min}\limits_{\boldsymbol{\varphi}, \boldsymbol{p}, \rho, \boldsymbol{f}, \boldsymbol{\varrho}}& \quad \frac{1}{K_{\mathtt{MEC}} }\sum_{k \in \mathcal{K}_{\mathtt{MEC}} } T_k^{'\mathtt{e2e},(i)}  \label{eq:Main_convexa} \\ 
		\st& \eqref{eq:Mainf}, \eqref{eq:Maing}, \eqref{eq:Mainh}, \eqref{eq:r_k_equi}, \eqref{eq:Te2e_SCA},\eqref{eq:SCA_energy_Const}, \eqref{eq:SCA_R_cm}, \eqref{eq:SCA_CRLB}   \label{eq:Main_convexb}.\quad\
	\end{IEEEeqnarray}
\end{subequations}
\begin{algorithm}[t]
\small
\begin{algorithmic}[1]
\protect\caption{Proposed SCA-based Iterative
Algorithm for Solving Problem \eqref{eq:Main}} \label{alg_1} \global\long\def\algorithmicrequire{\textbf{Initialization:}}	
      \REQUIRE Set $i: = 0$ and randomly generate feasible points $({\boldsymbol{\varphi }}^{(0)},\boldsymbol{p}^{(0)},\rho^{(0)},\boldsymbol{f}^{(0)},\boldsymbol{\varrho}^{(0)})$.
	   \REPEAT
	    \STATE Solve problem \eqref{eq:Main_convex} to obtain the  optimal solutions $({\boldsymbol{\varphi }}^{*},{\boldsymbol{p}}^{*}, \rho^{*}, \boldsymbol{f}^{*},\boldsymbol{\varrho}^{*})$;
	    \STATE Set $(\boldsymbol\varphi^{(i+1)},\boldsymbol p^{(i+1)},\rho^{(i+1)},\boldsymbol f^{(i+1)},\boldsymbol\varrho^{(i+1)})$ to the solution of \eqref{eq:Main_convex}.
	    
	    \STATE Set $i: = i + 1$;
	    
	    \UNTIL Convergence 
           \STATE \textbf{Output} $({\boldsymbol{\varphi }}^{*},{\boldsymbol{p}}^{*}, \rho^{*}, \boldsymbol{f}^{*},\boldsymbol{\varrho}^{*})$.
\end{algorithmic}
\end{algorithm}
\subsection{Algorithm and Complexity Analysis}\label{subsec_complexity_analysis}
The SCA-based solution for solving problem \eqref{eq:Main} is summarized in Algorithm~\ref{alg_1}. At each iteration $i$, the optimal solution of the convex program \eqref{eq:Main_convex} is used as the feasible point for the next iteration, and this process is repeated until convergence.
Note that the cubic terms in constraint \eqref{eq:SCA_energy_Const} can be exactly represented by a constant number of auxiliary variables and second-order cone (SOC) constraints. Thus, problem \eqref{eq:Main_convex} only includes {$4K_{\mathtt{MEC}} + K_{\mathtt{com}} +1$} scalar variables and {$4K_{\mathtt{MEC}} + 2K_{\mathtt{com}} +4$} linear, quadratic and SOC constraints. Thus, the computational complexity per iteration of Alg. \ref{alg_1} is {$\mathcal{O}\big( \sqrt{4K_{\mathtt{MEC}} + 2K_{\mathtt{com}} +4} (4K_{\mathtt{MEC}} + K_{\mathtt{com}} +1)^3 \big)$}.

\section{Numerical Results}\label{sec_NumericalResults}
 \subsection{Simulation Setting}\label{sec_simulation_setup}
 We consider a FD-mMIMO ISCC network in which the UEs are uniformly distributed within a cell of radius $1000$ meters (m), with the AP located at the center. No UE is located closer to the AP than $d_h = 100$ m. The target is located $400$ m from the AP at angles $(\theta, \phi) = (\frac{\pi}{8}, \frac{\pi}{4})$, while both internal leakage angles $(\theta_{\mathtt{SI}}, \phi_{\mathtt{SI}})$ are set to $\frac{\pi}{2}$. The large-scale fading coefficients are modeled as $\beta_k = z_k/ (d_k/d_h)^{\nu}$, where $z_k$ is a log-normal random variable with standard deviation $\sigma_{\mathtt{shadow}}=7$ dB, $d_k$ is the UE-AP distance, and $\nu = 3.2$ is the path loss exponent~\cite{Ngo2013}.  We assume $\alpha = \frac{1}{\beta_s \sqrt{2}} (1+j)$, where $\beta_s$ denotes the round-trip path loss between the AP and the target. We set $L = 30$ and $\text{SNR} = P_{\mathtt{AP}}^{\max}$ since the noise variances are normalized as $\sigma^2_c = 1$ \cite{Ngo2013}.  The RSI power factor is set to $\chi_{\mathtt{SI}} = -60$ dB (\textit{i.e.} $10^{-6}$) \cite{Bharadia_2013}. The Rician $\kappa$-factor is set to $\kappa_{\mathtt{SI}} = 5$ dB \cite{JSAC_FD18}. We set the symbol period $T_s = 1/B$, the pilot sequence length $\tau_p = K $ and the pilot transmit power $\vartheta_u = P_k^{\max}$. 

 The array response vectors \eqref{eq:resonponse_vector_h} and \eqref{eq:response_vector_v} are set to $N_{\mathtt{th}} = N_{\mathtt{tv}}$ $= \sqrt{N_{\mathtt{t}}} \in \mathbb{N}$ and $N_{\mathtt{rh}} = N_{\mathtt{rv}}$ $= \sqrt{N_{\mathtt{r}}} \in \mathbb{N}$. Other parameters are listed in Table~\ref{tab:Simulationparameter1}.
  \begin{table}[t]
\centering  
\captionof{table}{Simulation Parameters}
\label{tab:Simulationparameter1}
\scalebox{0.77}{
    \begin{tabular}{l|l}
        \hline
        Parameter & Value \\
        \hline\hline
        System bandwidth, $B$ & $10$ MHz \cite{Huynh_ISSC_2024} \\
        Number of UEs, $K_{\mathtt{MEC}}$ $\&$ $K_{\mathtt{com}}$ & $8$ $\&$ $8$ \cite{huynh_service_placement}\\
        UE's processing rate, $f_k^{\mathtt{ue}}, \forall k$ & $1$ GHz \cite{huynh_service_placement}\\
        UE's maximum transmit power, $P_k^{\max}$ & $23$ dBm \cite{huynh_service_placement}\\
        Size of task, $D_k, \forall k \in \mathcal{K}_{\mathtt{MEC}}$ & $1354$ bytes \cite{3GPP} \\
        Complexity of task, $C_k/D_k, \forall k \in \mathcal{K}_{\mathtt{MEC}}$ & $[200, 500]$ $\frac{\text{cycles}}{bytes}$ \\
        Task deadline, $T_k^{\max}$ & $3$ ms \cite{Bennis_2018} \\
        Minimum communication rate requirement, $R^{\min}$ & $1$ Mbps~\cite{Son_2026} \\
        ES's maximum computing capacity, $F_{\mathtt{es}}^{\max}$ & $20$ GHz \\
        UE's maximum energy consumption, $E_k^{\max }$ & $1$ J \\
        Effective capacitance coefficient, ${\mu _k}$ & ${10^{ - 27}}\,\frac{\mathrm{W}\cdot\mathrm{s}^3}{\mathrm{cycle}^3}$\\
        \hline		  				
    \end{tabular}
}
\end{table}

\subsection{Benchmark Schemes}\label{sec_benchmark_schemes}
 To evaluate the proposed optimization framework, we employ CVX to solve the convex program in \eqref{eq:Main_convex}. For comparison, we consider five benchmark approaches:
\begin{itemize}
    \item \textbf{Brute-Force search (BFS):} This benchmark searches candidate resource allocations at high computational cost. We note that its accuracy depends on the search resolution. 
    \item \textbf{Half-duplex ISCC (HD-ISCC):} We consider an orthogonal time-division multiplexing protocol. In Phase 2, to completely prevent SI, the temporal resource is equally partitioned ($50\%$-$50\%$) into two sequential sub-phases for uplink communications and radar sensing, respectively \cite{Dan_nguyen_2014,JSAC_FD18}. Consequently, the effective achievable rate of the users is strictly scaled by a factor of 0.5, given as $R_{k,\mathtt{HD}}^{\mathtt{cp}} = 0.5 B \log_2(1 + \frac{p_k (N_{\mathtt{r}} - K) \xi_k}{I^{\mathtt{er}}(\boldsymbol{p}) + \sigma_c^2}), \forall k \in \mathcal{K}_{\mathtt{MEC}}$ and  $R_{k',\mathtt{HD}}^{\mathtt{cm}} = 0.5 B \log_2(1 + \frac{p_{k'} (N_{\mathtt{r}} - K) \xi_{k'}}{I^{\mathtt{er}}(\boldsymbol{p}) + \sigma_c^2}), \forall k' \in \mathcal{K}_{\mathtt{com}}$. 
    Symmetrically, the sensing operation is entirely free from the multi-user uplink interference. Thus, the interference-plus-noise covariance matrix is simplified to $\boldsymbol{R}_{\mathtt{in}} = (\chi_{\mathtt{SI}} \rho + \sigma_c^2) \boldsymbol{I}_{N_{\mathtt{r}} L}$. Consequently, the total e2e latency is calculated as $T_{k,\mathtt{HD}}^{\mathtt{e2e}} = \max \big(t_{k}^{\mathtt{loc}} , t_p + t_{k,\mathtt{HD}}^{\mathtt{up}} + t_{k}^{\mathtt{es}}\big)$.
    \item \textbf{Equal power allocation among communication UEs (EqCM):} We evaluate two cases, with all communication UEs transmitting at either $P_k^{\max} /2$ or $P_k^{\max}$ to highlight the effectiveness of joint power optimization.
    \item \textbf{Equal power allocation among MEC UEs (EqMEC):} It fixes the transmit power of all MEC UEs to $P_k^{\max} /2$ and $P_k^{\max}$ to evaluate the impact of optimizing MEC power.
    \item  \textbf{Fixed offloading fractions $(20\%, 50\%, 70\%)$:} These schemes fix the fractions of tasks offloaded by MEC UEs to the ES to study the effect of offloading decisions in the FD-mMIMO ISCC system.
    
\end{itemize}

\subsection{Numerical Results and Discussion}\label{sec_numerical_results_dicussions}
\begin{figure}[!htb]
    \centering
    \vspace{-2mm}
    \subfigure[Constant transmit power regime ($p_k = 1, \forall k \in \mathcal{K} $).]{
        \includegraphics[width = 3in]{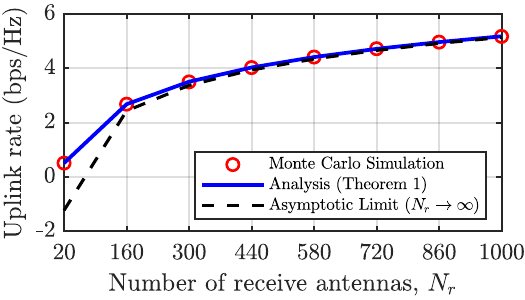}
        \label{fig:uplink_rate_fix_power}
    }
    \subfigure[Power scaling regime ($p_k = P_k^{\mathtt{ref}}/N_{\mathtt{r}}$), showing MUI elimination and convergence to the radar interference ceiling.]{
        \includegraphics[width =3in]{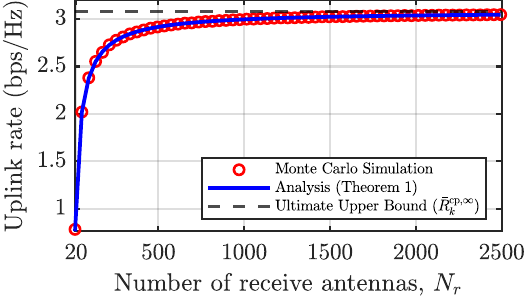}
        \label{fig:uplink_rate_scaling_law_power}
    }
      \vspace{-2mm}
    \caption{\small Uplink offloading rate of the UE $k \in \mathcal{K}_{\mathtt{MEC}}$ versus $N_{\mathtt{r}}$ with $K_{\mathtt{MEC}} = 4$, $K_{\mathtt{com}} = 4$, $N_{\mathtt{th}} = 8$, $N_{\mathtt{tv}} = 8$. }
    \label{fig:verify_theorem1}
\end{figure}
\textit{1) Validation of asymptotic performance and power scaling:} We validate the analytical uplink rates \eqref{eq:up_link_rate} against Monte Carlo (MC) simulations \eqref{eq:Rate_MC} over $5 \times 10^3$ channel realizations. As shown in Fig.~\ref{fig:verify_theorem1}, the analytical lower bounds tightly match the MC results across all antenna configurations, validating our derivations. Fig.~\ref{fig:uplink_rate_fix_power} shows that the achievable rate increases with $N_{\mathtt{r}}$ due to the massive array gain, closely tracking the derived asymptotic limit under constant transmit powers ($p_k = 1, \forall k \in \mathcal{K}$). Conversely, under the power scaling law ($P_k^{\mathtt{ref}} = 150$) \cite{Ngo2013}, Fig.~\ref{fig:uplink_rate_scaling_law_power} shows that both the analytical and MC rates converge to the finite saturation ceiling ($\bar{R}_k^{\mathtt{cp}, \infty}$) defined in \eqref{eq:ultimate_upper_bound} near $N_{\mathtt{r}} = 1000$. Although scaling $N_{\mathtt{r}}$ conserves user energy, the uplink rate ultimately saturates. Under this power-scaling regime, radar interference remains as $N_{\mathtt{r}}$ increases, motivating joint optimization of radar and user powers.

\begin{figure}[!htb]
    \centering
    \vspace{-2mm}
    \subfigure[Interference mitigation via $N_{\mathtt{t}}$ with $N_{\mathtt{r}} = 16$]{
        \includegraphics[width = 3in]{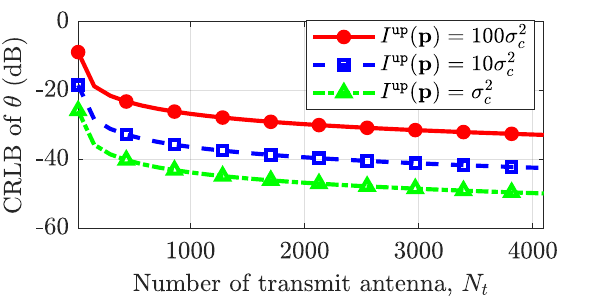}
        \label{fig:CRLB_vs_Nt}
    }
    \subfigure[Error decay rate comparison with varying $N_{\mathtt{t}}$ or $N_{\mathtt{r}}$]{
        \includegraphics[width =3in]{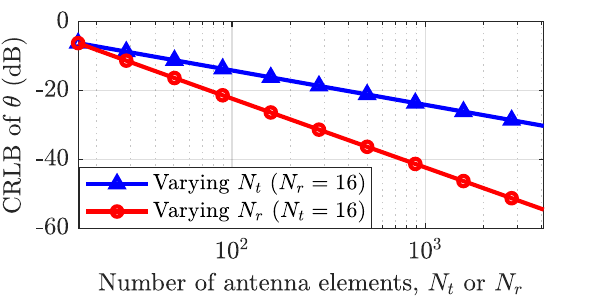}
        \label{fig:decay_rate_vs_Nt_Nr}
    }
      \vspace{-2mm}
    \caption{\small Asymptotic sensing limits ($L = 30$, $\sigma_c^2 = 1$, $\chi_{\mathtt{SI}} = 10^{-6}$). }
    \label{fig:verify_sensing_limits}
\end{figure}
Fig.~\ref{fig:verify_sensing_limits} validates the asymptotic limits from Remark~\ref{remark_Nt_CRLB}. In Fig.~\ref{fig:CRLB_vs_Nt}, the CRLB vanishes as $N_{\mathtt{t}} \to \infty$, consistent with the transmit-array gain in the sensing bounds. Furthermore, Fig.~\ref{fig:decay_rate_vs_Nt_Nr} confirms that scaling $N_{\mathtt{r}}$ yields a much steeper error decay ($-20$ dB/decade) than $N_{\mathtt{t}}$ ($-10$ dB/decade). We observe that under the stated square-array assumptions, increasing the receive-array size reduces the angular CRLBs faster than increasing the transmit-array size.

\begin{figure}[!htb]
	\centering
    \includegraphics[width=0.9\columnwidth,trim={0cm 0.0cm 0cm 0.0cm}]{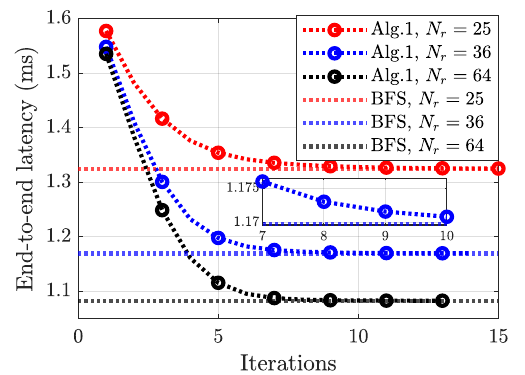}
    \vspace{-2mm}
	\caption{\small Convergence of Algorithm \ref{alg_1} with $\Phi_{\theta}^0 = \Phi_{\phi}^0  = -30$ dB, $N_{\mathtt{t}} = 225$, $K_{\mathtt{MEC}} = 8$, $K_{\mathtt{com}} = 8$, $P_{\mathtt{AP}}^{\max} = 43 $ dBm and $L = 30$.}
	\label{fig:Convergence}
\end{figure} 
\textit{2) Convergence analysis of Algorithm \ref{alg_1}:} Fig.~\ref{fig:Convergence} illustrates the convergence of Algorithm~\ref{alg_1} and the BFS algorithm when solving problem~\eqref{eq:Main} for $N_{\mathtt{r}} \in \{25, 36, 64\}$. As observed, Algorithm~\ref{alg_1} converges within approximately seven iterations across all $N_{\mathtt{r}}$ cases, yielding a marginal objective gap (around $4 \times 10^{-3}$) compared to BFS. Additionally, the average e2e latency of MEC UEs decreases as $N_{\mathtt{r}}$ increases from $25$ to $64$. 

\begin{figure}[!htb]
	\centering
    \includegraphics[width=0.9\columnwidth,trim={0cm 0.0cm 0cm 0.0cm}]{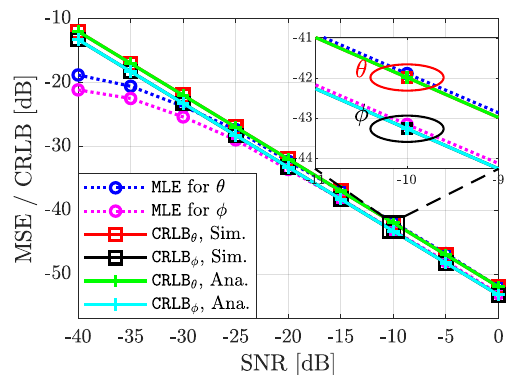}
    \vspace{-2mm}
	\caption{\small CRLB  and MSE of the MLE for $\theta$ and $\phi$ with $N_{\mathtt{t}} = 64$, $N_{\mathtt{r}} = 64$, $K_{\mathtt{MEC}} = 8$, $K_{\mathtt{com}} = 8$, and $L = 30$. The MLE grid-search step size is $\pi/256$.}
	\label{fig:CRLB_vs_MLE}
\end{figure}
\textit{3) Offloading and sensing performance:} We first evaluate the theoretical CRLB derived in Theorem~\ref{theorem_CRLB} and Remark~\ref{remark_tractability_CRLB}. Fig.~\ref{fig:CRLB_vs_MLE} compares the analytical CRLB for $\theta$ and $\phi$ with MC simulations based on \eqref{eq:CRLB_theta} and \eqref{eq:CRLB_phi}. To further validate these theoretical bounds, we compute the MSE of the MLE in \eqref{eq:MLE_w} using an exhaustive search over a fine grid. Following \cite{Liu2021CRB}, the effective sensing SNR in the tracking mode is defined as $\frac{p_{\mathtt{t}} L |\alpha|^2}{ I^{\mathtt{up}}(\boldsymbol{p}) + \chi_{\mathtt{SI}} \rho + \sigma_c^2}$, where $p_{\mathtt{t}} = \rho$. As observed in Fig.~\ref{fig:CRLB_vs_MLE}, the analytical and MC CRLBs provide a tight lower bound for the MLE's MSE. This comparison evaluates the bounds under the adopted sensing model.

\begin{figure}[t]
    \centering
    \vspace{-2mm}
    \subfigure[Average e2e latency and offloading fraction versus CRLB threshold.]{
        \includegraphics[width = 3.0in]{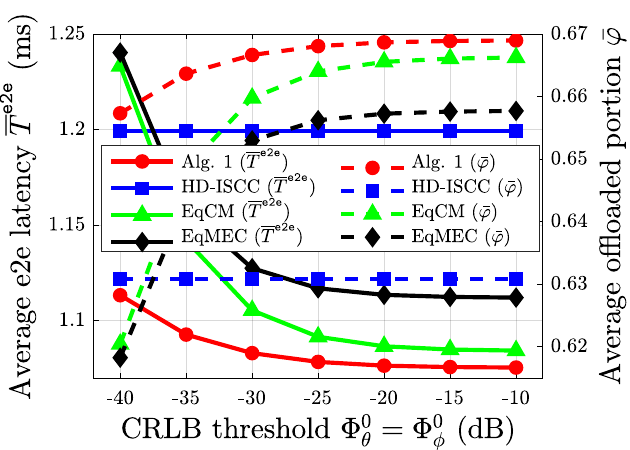}
        \label{fig:latency_offloading_vs_CRLB}
    }
    \subfigure[Average offloading fraction and AP transmit power versus CRLB threshold.]{
        \includegraphics[width =3.0in]{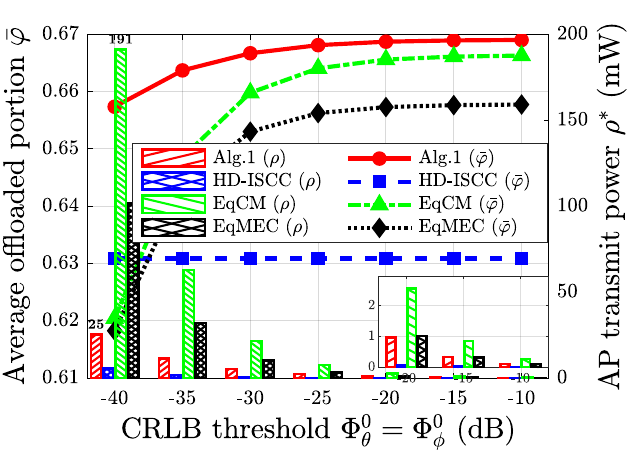}
        \label{fig:offloading_APpower_vs_CRLB}
    }
      \vspace{-2mm}
    \caption{\small Effect of CRLB threshold on computing performance and AP power with  $N_{\mathtt{t}} = 225$, $K_{\mathtt{MEC}} = 8$, $K_{\mathtt{com}} = 8$ and $L = 30$.}
    \label{fig:effect_CRLB_threshold}
\end{figure}
Fig.~\ref{fig:effect_CRLB_threshold} shows how the sensing threshold affects latency, offloading, and radar power. Relaxing $\Phi_\psi^0$ from $-40$ to $-10$ dB reduces the required radar power and its interference to uplink reception, allowing more offloading and lower latency. Algorithm~\ref{alg_1} achieves the lowest latency among the compared schemes. HD-ISCC requires less radar power because sensing is free from uplink interference, but its reduced communication time limits offloading. At $-40$ dB, the proposed allocation reduces radar power by approximately $86\%$ relative to EqCM. It is clear that fixing MEC or communication-user powers limits the ability to balance sensing and offloading requirements.

\begin{figure}[t]
    \centering
    \vspace{-2mm}
    \subfigure[\small $\rho^*$ and $\overline{T}^{\mathtt{e2e}}$ versus $N_{\mathtt{r}}$ with different $\Phi_{\psi}^0$ of Algorithm \ref{alg_1}]{
        \includegraphics[width = 3in]{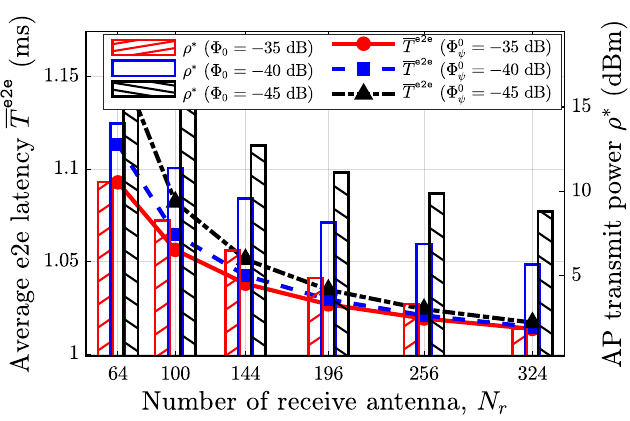}
        \label{fig:APpower_latency_vs_Nr}
    }
    \subfigure[\small Average UE energy consumption and e2e latency versus $N_{\mathtt{r}}$ with EqCM$(P_k^{\max})$ and EqMEC$(P_k^{\max})$.] {
        \includegraphics[width =3in]{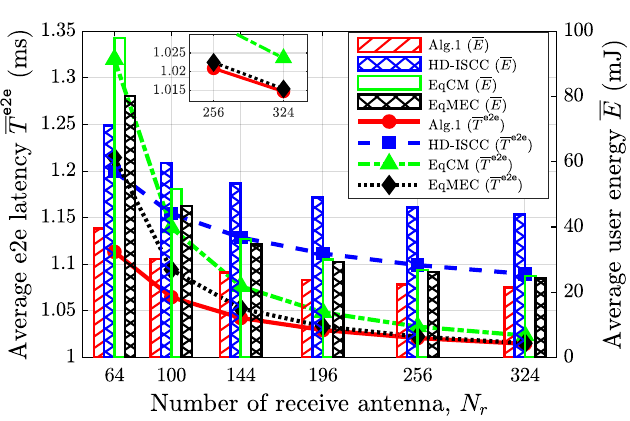}
        \label{fig:Energy_latency_vs_Nr}
    }
      \vspace{-2mm}
    \caption{\small Impact of $N_{\mathtt{r}}$ on AP power, user e2e latency, and energy ($N_{\mathtt{t}} = 225$, $K_{\mathtt{MEC}} = K_{\mathtt{com}} = 8$, $L = 30$, $\Phi_{\psi}^0 = -40$ dB).}
    \label{fig:effect_Nr}
\end{figure}

Fig.~\ref{fig:effect_Nr} shows that larger receive arrays reduce latency, radar power, and UE energy consumption. In Fig.~\ref{fig:APpower_latency_vs_Nr}, improved sensing resolution reduces the radar power needed to meet each CRLB threshold, while higher uplink rates accelerate offloading. The latency differences between sensing thresholds decrease with $N_{\mathtt{r}}$. These finite-array results complement the fixed-radar-power scaling analysis in Remark~\ref{remark_dof_limits}, rather than representing the same asymptotic regime. Fig.~\ref{fig:Energy_latency_vs_Nr} shows that joint allocation retains an advantage over fixed-power schemes, although the latency gap narrows as the array grows.
\begin{figure}[!htb]
	\centering
    \includegraphics[width=0.9\columnwidth,trim={0cm 0.0cm 0cm 0.0cm}]{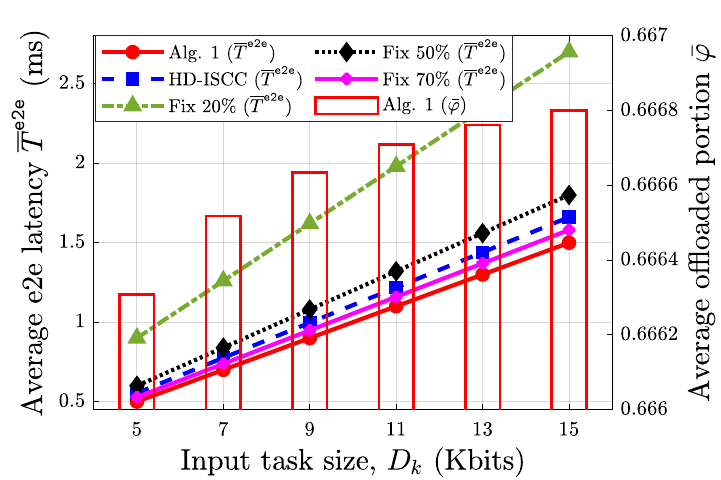}
    \vspace{-2mm}
	\caption{\small Effect of task size $D_k$ on average e2e latency with $\Phi_{\theta}^0 = \Phi_{\phi}^0  = -30$ dB, $N_{\mathtt{t}} = 225$, $N_{\mathtt{r}} = 64$, $K_{\mathtt{MEC}} = 8$, $K_{\mathtt{com}} = 8$, $P_{\mathtt{AP}}^{\max} = 43 $ dBm and $L = 30$.}
	\label{fig:latency_vs_DK}
\end{figure}

Fig.~\ref{fig:latency_vs_DK} investigates the impact of the input task size $D_k$ on the ISCC system across fixed offloading fractions ($20\%$, $50\%$, and $70\%$). As observed, the latency of all schemes increases as $D_k$ grows from $5$ to $15$ kbits, though higher offloading fractions generally help reduce this delay. Algorithm \ref{alg_1} achieves the lowest latency with an average offloading fraction of $66.68\%$, below the fixed $70\%$ benchmark. Thus, offloading a larger task fraction does not necessarily reduce latency.

\begin{figure}[t]
    \centering
    \vspace{-2mm}
    \subfigure[\small Effect of the total number of UEs $K$  on average e2e latency]{
        \includegraphics[width = 3in]{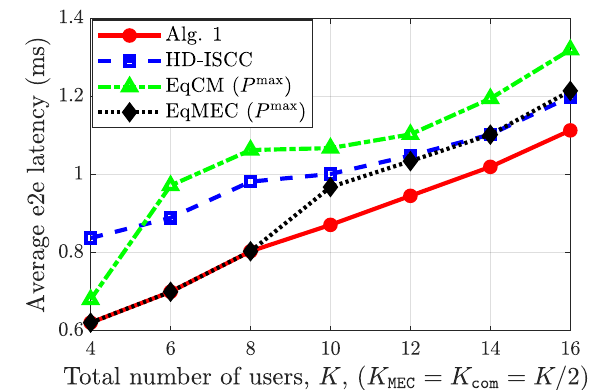}
        \label{fig:latency_vs_num_UE}
    }
    \subfigure[\small Cumulative probability of latency with $K_{\mathtt{MEC}} = 8, K_{\mathtt{com}} = 8$] {
        \includegraphics[width =3in]{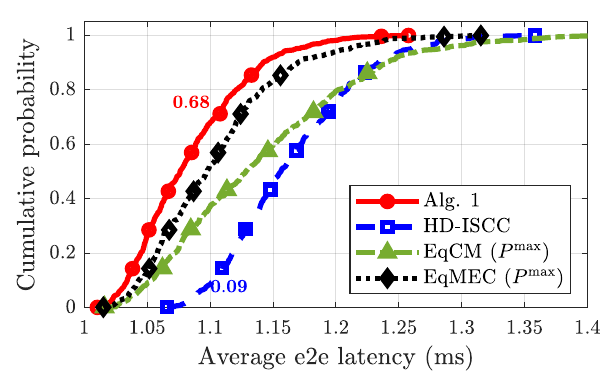}
        \label{fig:latnecy_CDF}
    }
      \vspace{-2mm}
    \caption{\small Computing performance with $\Phi_{\theta}^0 = \Phi_{\phi}^0  = -40$ dB, $N_{\mathtt{t}} = 225$, $N_{\mathtt{r}} = 64$, $P_{\mathtt{AP}}^{\max} = 43 $ dBm and $L = 30$. }
    \label{fig:effect_number_UE_CDF}
\end{figure}

Fig.~\ref{fig:effect_number_UE_CDF} shows that latency increases with the number of users. The proposed allocation achieves lower latency than the benchmarks, with a larger advantage over EqMEC as user loading increases. In Fig.~\ref{fig:latnecy_CDF}, the probability of latency below $1.1$ ms is approximately $68\%$ for Algorithm~\ref{alg_1} and $9\%$ for HD-ISCC. These results illustrate the benefit of adapting transmission and computing resources to user loading. Increasing $K$ reduces the available ZF degrees of freedom $N_{\mathtt{r}}-K$ and increases competition for edge CPU resources. With imperfect CSI, fixed user powers also increase residual interference. These effects explain why the gap between joint allocation and fixed-power schemes becomes more visible at higher user loads. 

\section{Conclusion}\label{sec_conclusion}
We investigated sensing-aware latency minimization in FD massive MIMO ISCC networks. In particular, the uplink rate and angular CRLB analysis characterized the coupling between radar transmission, computation offloading, and sensing accuracy. The scaling analysis further distinguished the sensing gains of transmit and receive arrays and identified a finite uplink rate limit when user powers decrease inversely with receive-array size. This result motivates coordinating radar power with user-power reduction as arrays grow. Based on these metrics, we developed an SCA-based algorithm that jointly allocates offloading fractions, user and radar powers, and edge CPU resources. Numerical results showed lower latency than the considered fixed-allocation and HD benchmarks, together with an approximately $86\%$ radar-power reduction relative to EqCM at a $-40$ dB CRLB threshold. The results support joint allocation of radar power, user powers, and edge CPU resources to reduce latency under sensing constraints. The sensing analysis relies on the adopted interference approximation and beam-alignment assumptions. Future work will extend the framework to spatially colored interference and uncertain target directions.

\appendices
\renewcommand{\thesectiondis}[2]{\Alph{section}:}
\section{Proof of Theorem~\ref{theorem_uplink_rate}}\label{proof_uplink_rate}

We apply the UatF bound to evaluate the independent interference expectations in Eq. \eqref{eq:Rate_MC}.

    \noindent1) \textit{Desired signal power}: Since  $\mathbb{E}[|s_k|^2] = 1$, we obtain $\mathbb{E} \left[ |\sqrt{p_k} s_k|^2 \right] = p_k\mathbb{E} \left[ | s_k|^2 \right] = p_k$.

\noindent 2) \textit{Inter-user interference power}: Since the error vectors $\boldsymbol{e}_i \sim \mathcal{CN}(\boldsymbol{0}, \epsilon_i\boldsymbol{I}_{N_{\mathtt{r}}})$ ($i \in \{j, k'\}$) are independent of $\boldsymbol{u}_k$, we have $\mathbb{E}[|\boldsymbol{u}_k^{\mathtt{H}} \boldsymbol{e}_i|^2] = \epsilon_i \mathbb{E}[\|\boldsymbol{u}_k\|^2]$. Summing over all interferers yields $\mathbb{E}[I_k^{\mathtt{ce}}] = \big( \sum_{j \in \mathcal{K}_{\mathtt{MEC}}} p_j \epsilon_j + \sum_{k' \in \mathcal{K}_{\mathtt{com}}} p_{k'} \epsilon_{k'} \big) \mathbb{E}[\|\boldsymbol{u}_k\|^2] \triangleq I^{\mathtt{er}}(\boldsymbol{p}) \mathbb{E}[\|\boldsymbol{u}_k\|^2]$.

\noindent 3) \textit{Cross-interference power}: Since the entries of $\boldsymbol{G}_{\mathtt{SI}}^{\mathtt{NLoS}}$ are i.i.d. $\mathcal{CN}(0,1)$, we have $\mathbb{E}[\vert\boldsymbol{u}_k^{\mathtt{H}} (\boldsymbol{G}_{\mathtt{SI}}^{\mathtt{NLoS}})^{\mathtt{H}} \tilde{\boldsymbol{w}}\vert^2] = \mathbb{E}[\Vert\boldsymbol{u}_k\Vert^2]$. Thus, we obtain:
    \begin{align}\label{eq:P_sen_decomposed}
        &\mathbb{E}[I_k^{\mathtt{ec}}] + \mathbb{E}[I_k^{\mathtt{rsi}}] = \rho \vert\alpha\vert^2 \mathbb{E} \left[ \vert \boldsymbol{u}_k^{\mathtt{H}} \boldsymbol{G}(\theta, \phi) \tilde{\boldsymbol{w}} \vert^2 \right] + \frac{\rho \chi_{\mathtt{SI}}\kappa_{\mathtt{SI}}}{\kappa_{\mathtt{SI}}+1}  \nonumber\\ 
        & \times  \mathbb{E} \left[ \vert \boldsymbol{u}_k^{\mathtt{H}} (\boldsymbol{G}_{\mathtt{SI}}^{\mathtt{LoS}})^{\mathtt{H}} \tilde{\boldsymbol{w}} \vert^2 \right] + \frac{\rho \chi_{\mathtt{SI}}}{\kappa_{\mathtt{SI}}+1} \mathbb{E}[\Vert\boldsymbol{u}_k\Vert^2].
    \end{align}
Let $\boldsymbol{u}_k = \Vert\boldsymbol{u}_k\Vert \bar{\boldsymbol{u}}_k$. For large $N_{\mathtt{r}}$, the statistically isotropic direction $\bar{\boldsymbol{u}}_k$ yields $\mathbb{E}[\bar{\boldsymbol{u}}_k \bar{\boldsymbol{u}}_k^{\mathtt{H}}] = \frac{1}{N_{\mathtt{r}}}\boldsymbol{I}_{N_{\mathtt{r}}}$. Applying this with $\tilde{\boldsymbol{w}} = \frac{1}{\sqrt{N_{\mathtt{t}}}}\boldsymbol{a}_{\mathtt{t}}(\theta, \phi)$ gives:
   \begin{align}\label{eq:echo_power}
    &\mathbb{E} \left[ \vert \boldsymbol{u}_k^{\mathtt{H}} \boldsymbol{G}(\theta, \phi) \tilde{\boldsymbol{w}} \vert^2 \right] \nonumber\\ 
    &= \mathbb{E}[\Vert\boldsymbol{u}_k\Vert^2] \, \tilde{\boldsymbol{w}}^{\mathtt{H}} \boldsymbol{G}^{\mathtt{H}}(\theta, \phi) \mathbb{E}[\bar{\boldsymbol{u}}_k \bar{\boldsymbol{u}}_k^{\mathtt{H}}] \boldsymbol{G}(\theta, \phi) \tilde{\boldsymbol{w}}  \nonumber\\
    &= \mathbb{E}[\Vert\boldsymbol{u}_k\Vert^2] \, \frac{1}{N_{\mathtt{r}}} \Vert \sqrt{N_{\mathtt{t}}} \boldsymbol{a}_{\mathtt{r}}(\theta, \phi) \Vert^2 = \mathbb{E}[\Vert\boldsymbol{u}_k\Vert^2] \, N_{\mathtt{t}}.
\end{align}
   Similarly, we have $\mathbb{E} \left[ \vert \boldsymbol{u}_k^{\mathtt{H}} (\boldsymbol{G}_{\mathtt{SI}}^{\mathtt{LoS}})^{\mathtt{H}} \tilde{\boldsymbol{w}} \vert^2 \right] \!\!\!= \!\!\!\mathbb{E}[\Vert\boldsymbol{u}_k\Vert^2] \, \mathbb{E} \left[ \vert \bar{\boldsymbol{u}}_k^{\mathtt{H}} (\boldsymbol{G}_{\mathtt{SI}}^{\mathtt{LoS}})^{\mathtt{H}} \tilde{\boldsymbol{w}} \vert^2 \right].$
   Since  $\boldsymbol{G}_{\mathtt{SI}}^{\mathtt{LoS}} \!\!\!= \!\!\!\boldsymbol{a}_{\mathtt{t}}(\theta_{\mathtt{SI}}, \phi_{\mathtt{SI}}) \boldsymbol{a}_{\mathtt{r}}^{\mathtt{H}}(\theta_{\mathtt{SI}}, \phi_{\mathtt{SI}})$, the isotropic combining property yields $ \mathbb{E} \left[ \left\vert{} \bar{\boldsymbol{u}}_k^{\mathtt{H}} (\boldsymbol{G}_{\mathtt{SI}}^{\mathtt{LoS}})^{\mathtt{H}} \tilde{\boldsymbol{w}} \right\vert{}^2 \right] = \frac{1}{N_{\mathtt{r}}} \left\Vert{} (\boldsymbol{G}_{\mathtt{SI}}^{\mathtt{LoS}})^{\mathtt{H}} \tilde{\boldsymbol{w}} \right\Vert{}^2 = \frac{1}{N_{\mathtt{r}}} \left\vert{} \boldsymbol{a}_{\mathtt{t}}^{\mathtt{H}}(\theta_{\mathtt{SI}}, \phi_{\mathtt{SI}}) \tilde{\boldsymbol{w}} \right\vert{}^2 \left\Vert{} \boldsymbol{a}_{\mathtt{r}}(\theta_{\mathtt{SI}}, \phi_{\mathtt{SI}}) \right\Vert{}^2$. 
   
   From $\tilde{\boldsymbol{w}} = \frac{1}{\sqrt{N_{\mathtt{t}}}}\boldsymbol{a}_{\mathtt{t}}(\theta, \phi)$ and applying Kronecker decomposition with $\Vert\boldsymbol{a}_{\mathtt{r}}(\theta_{\mathtt{SI}}, \phi_{\mathtt{SI}})\Vert^2 =$ $ N_{\mathtt{r}}$ yields the Dirichlet kernels $\left\vert{} \boldsymbol{a}_{\mathtt{t}}^{\mathtt{H}}(\theta_{\mathtt{SI}}, \phi_{\mathtt{SI}}) \tilde{\boldsymbol{w}} \right\vert{}^2 $ $= \frac{1}{N_{\mathtt{t}}} \left\vert{} \boldsymbol{a}_{\mathtt{t}}^{\mathtt{H}}(\theta_{\mathtt{SI}}, \phi_{\mathtt{SI}}) \boldsymbol{a}_{\mathtt{t}}(\theta, \phi) \right\vert{}^2 = \frac{1}{N_{\mathtt{t}} } \Big| \Big( \sum_{m=-\frac{N_{\mathtt{t},\mathtt{h}}-1}{2}}^{\frac{N_{\mathtt{t},\mathtt{h}}-1}{2}} e^{-j \pi m (\sin\theta_{\mathtt{SI}} \sin\phi_{\mathtt{SI}} - \sin\theta \sin\phi)} \Big)  \\ \Big( \sum_{n=-\frac{N_{\mathtt{t},\mathtt{v}}-1}{2}}^{\frac{N_{\mathtt{t},\mathtt{v}}-1}{2}} e^{-j \pi n (\cos\phi_{\mathtt{SI}} - \cos\phi)} \Big) \Big|^2 = \frac{\mathcal{D}_{\mathtt{H}}(\theta, \phi) \mathcal{D}_{\mathtt{V}}(\phi)}{N_{\mathtt{t}}}$. Consequently, we obtain $\mathbb{E} \left[ \vert \boldsymbol{u}_k^{\mathtt{H}} (\boldsymbol{G}_{\mathtt{SI}}^{\mathtt{LoS}})^{\mathtt{H}} \tilde{\boldsymbol{w}} \vert^2 \right] =  \frac{\mathbb{E}[\Vert\boldsymbol{u}_k\Vert^2]}{N_{\mathtt{t}}} \mathcal{D}_{\mathtt{H}}(\theta, \phi) \mathcal{D}_{\mathtt{V}}(\phi)$. Substituting this result and \eqref{eq:echo_power} into \eqref{eq:P_sen_decomposed}, we obtain $\mathbb{E}[I_k^{\mathtt{ec}}] + \mathbb{E}[I_k^{\mathtt{rsi}}] = I^\mathtt{se}(\rho) \mathbb{E}[\Vert{}\boldsymbol{u}_k\Vert{}^2]$.
   
   Finally, evaluating the AWGN as $\mathbb{E}[\vert{}\boldsymbol{u}_k^{\mathtt{H}} \boldsymbol{n}_c\vert{}^2] = \sigma_c^2 \mathbb{E}[\Vert{}\boldsymbol{u}_k\Vert{}^2]$ and combining all evaluated expectations into \eqref{eq:Rate_MC} yields $R_k^\mathtt{cp}(\boldsymbol{p}, \rho) = B \log_2 \big( 1 + \frac{p_k \left( 1 / \mathbb{E}[\Vert{}\boldsymbol{u}_k\Vert{}^2] \right)}{I^{\mathtt{er}}(\boldsymbol{p}) + I^\mathtt{se}(\rho) + \sigma_c^2} \big)$. By the ZF inverse Wishart property $1/\mathbb{E}[\Vert\boldsymbol{u}_k\Vert^2] = (N_{\mathtt{r}} - K)\xi_k$, we obtain $R_k^\mathtt{cp}(\boldsymbol{p}, \rho)$ as given in \eqref{eq:up_link_rate}.

\section{Proof of Remark~\ref{remark_dof_limits}}\label{proof_remark_dof}
For fixed $K$, $N_{\mathtt{t}}$, pilot powers, and radar power, the SINR in \eqref{eq:up_link_rate} is
\begin{equation}
\gamma_k(N_{\mathtt{r}})=g_k(N_{\mathtt{r}}-K),\ \text{with}\
 g_k=\frac{p_k\xi_k}{I^{\mathtt{er}}(\boldsymbol p)+I^{\mathtt{se}}(\rho)+\sigma_c^2}.
\end{equation}
When user powers are fixed, $g_k$ is independent of $N_{\mathtt{r}}$. Hence, $\gamma_k(N_{\mathtt{r}})/N_{\mathtt{r}}\to g_k$ and
\begin{equation}
R_k^{\mathtt{cp}}=B\log_2 N_{\mathtt{r}}+B\log_2 g_k+o(1).
\end{equation}
Next, let $p_u=P_u^{\mathtt{ref}}/N_{\mathtt{r}}$ for every $u\in\mathcal K$, keeping the pilot powers fixed so that $\xi_k$ does not change. The desired-signal term satisfies
\begin{equation}
p_k(N_{\mathtt{r}}-K)\xi_k=P_k^{\mathtt{ref}}\xi_k\left(1-\frac{K}{N_{\mathtt{r}}}\right)
\longrightarrow P_k^{\mathtt{ref}}\xi_k.
\end{equation}
The channel-error interference becomes $I^{\mathtt{er}}(\boldsymbol p)=N_{\mathtt{r}}^{-1}\sum_{u\in\mathcal K}P_u^{\mathtt{ref}}\epsilon_u$ and tends to zero. Since radar interference and receiver noise remain fixed, the SINR converges to $P_k^{\mathtt{ref}}\xi_k/[I^{\mathtt{se}}(\rho)+\sigma_c^2]$, giving \eqref{eq:ultimate_upper_bound}. This limit applies to the stated power-scaling regime; it does not imply rate saturation when user powers remain fixed.

\section{Proof of Lemma~\ref{LM_lemma1}} \label{app:proofoflemma1}
With independent, zero-mean symbols ($\mathbb{E}[|s_k(l)|^2] = \mathbb{E}[|s_{k'}(l)|^2] = \mathbb{E}[|s_{\mathrm{rad}}(l)|^2] = 1$), all cross-correlations vanish, yielding the snapshot covariance:
\begin{align}\label{eq:R_in_snapshot_decomposed}
    &\boldsymbol{R}_{\mathrm{in}}(l) = \mathbb{E}[\boldsymbol{n}_{\mathrm{eq}}(l) \boldsymbol{n}_{\mathrm{eq}}^{\mathtt{H}}(l)]    \nonumber\\ 
    &  = \!\!\! \! \sum_{k \in \mathcal{K}_{\mathtt{MEC}}}\!\! p_k \mathbb{E} \left[ \boldsymbol{h}_k \boldsymbol{h}_k^{\mathtt{H}} \right]  + \!\!\!\!\sum_{k' \in \mathcal{K}_{\mathtt{com}}}\!\!\! p_{k'} \mathbb{E} \left[ \boldsymbol{h}_{k'} \boldsymbol{h}_{k'}^{\mathtt{H}} \right] \nonumber\\ 
    & + \underbrace{\chi_{\mathtt{SI}} \mathbb{E} \left[ \boldsymbol{G}_{\mathtt{SI}}^{\mathtt{H}} \boldsymbol{w} \boldsymbol{w}^{\mathtt{H}} \boldsymbol{G}_{\mathtt{SI}} \right]}_ {\boldsymbol{R}_{\mathtt{SI}}(l)} + \underbrace{\mathbb{E} \left[ \boldsymbol{n}_c(l)\boldsymbol{n}_c^{\mathtt{H}}(l) \right]}_{\sigma_c^2 \boldsymbol{I}_{N_{\mathtt{r}}}}.
\end{align}
Applying $\mathbb{E} [ \boldsymbol{h}_u \boldsymbol{h}_u^{\mathtt{H}} ] = \beta_u \mathbb{E} [ \bar{\boldsymbol{h}}_u \bar{\boldsymbol{h}}_u^{\mathtt{H}} ] = \beta_u \boldsymbol{I}_{N_{\mathtt{r}}}$ yields:
\begin{align}\label{eq:calculate_I_up}
   \!\!\! \! \sum_{k \in \mathcal{K}_{\mathtt{MEC}}}\!\! p_k \mathbb{E} \left[ \boldsymbol{h}_k \boldsymbol{h}_k^{\mathtt{H}} \right] + \!\!\!\!\sum_{k' \in \mathcal{K}_{\mathtt{com}}}\!\!\! p_{k'} \mathbb{E} \left[ \boldsymbol{h}_{k'} \boldsymbol{h}_{k'}^{\mathtt{H}} \right] =  I^{\mathtt{up}}(\boldsymbol{p})\boldsymbol{I}_{N_{\mathtt{r}}}.
\end{align}
 Substituting \eqref{eq:G_si} into $\boldsymbol{R}_{\mathtt{SI}}(l)$, the cross-terms vanish due to the i.i.d. $\mathcal{CN}(0, 1)$ entries of $\boldsymbol{G}_{\mathtt{SI}}^{\mathtt{NLoS}}$, and thus
\begin{align}\label{eq:R_SI_decoupled}
    \boldsymbol{R}_{\mathtt{SI}}(l) &= \underbrace{ \frac{\chi_{\mathtt{SI}} \kappa_{\mathtt{SI}}}{\kappa_{\mathtt{SI}}+1} \mathbb{E}\left[ (\boldsymbol{G}_{\mathtt{SI}}^{\mathtt{LoS}})^{\mathtt{H}} \boldsymbol{w} \boldsymbol{w}^{\mathtt{H}} \boldsymbol{G}_{\mathtt{SI}}^{\mathtt{LoS}} \right] }_{\text{Deterministic Colored LoS Component}} \\ \nonumber
    &\quad + \underbrace{ \frac{\chi_{\mathtt{SI}}}{\kappa_{\mathtt{SI}}+1} \mathbb{E} \left[ (\boldsymbol{G}_{\mathtt{SI}}^{\mathtt{NLoS}})^{\mathtt{H}} \boldsymbol{w} \boldsymbol{w}^{\mathtt{H}} \boldsymbol{G}_{\mathtt{SI}}^{\mathtt{NLoS}} \right] }_{\text{Scattered NLoS Component}}.
\end{align}
Applying $\mathbb{E}[\boldsymbol{H}^{\mathtt{H}} \boldsymbol{A} \boldsymbol{H}] = \text{Tr}(\boldsymbol{A})\boldsymbol{I}$ to the scattered NLoS component yields:
\begin{align}\label{eq:NLOS_component}
    \mathbb{E} \left[ (\boldsymbol{G}_{\mathtt{SI}}^{\mathtt{NLoS}})^{\mathtt{H}} \boldsymbol{w} \boldsymbol{w}^{\mathtt{H}} \boldsymbol{G}_{\mathtt{SI}}^{\mathtt{NLoS}} \right] = \|\boldsymbol{w}\|^2 \boldsymbol{I}_{N_{\mathtt{r}}} = \rho \boldsymbol{I}_{N_{\mathtt{r}}}.
\end{align}
Since the exact FIM for colored LoS interference is intractable \cite{Liu2021CRB}, we bound it as spatially white Gaussian noise \cite{Liu2021CRB,Li2007MIMORadar}. From $\boldsymbol{G}_{\mathtt{SI}}^{\mathtt{LoS}} = \boldsymbol{a}_{\mathtt{t}}(\theta_{\mathtt{SI}}, \phi_{\mathtt{SI}}) \boldsymbol{a}_{\mathtt{r}}^{\mathtt{H}}(\theta_{\mathtt{SI}}, \phi_{\mathtt{SI}})$, we have $\mathbb{E}[ (\boldsymbol{G}_{\mathtt{SI}}^{\mathtt{LoS}})^{\mathtt{H}} \boldsymbol{w} \boldsymbol{w}^{\mathtt{H}} \boldsymbol{G}_{\mathtt{SI}}^{\mathtt{LoS}} ] = |\boldsymbol{a}_{\mathtt{t}}^{\mathtt{H}} \boldsymbol{w}|^2 \boldsymbol{a}_{\mathtt{r}} \boldsymbol{a}_{\mathtt{r}}^{\mathtt{H}}$. Applying the isotropic approximation ($\boldsymbol{a}_{\mathtt{r}} \boldsymbol{a}_{\mathtt{r}}^{\mathtt{H}} \approx \boldsymbol{I}_{N_{\mathtt{r}}}$) and uncorrelated leakage assumption ($|\boldsymbol{a}_{\mathtt{t}}^{\mathtt{H}} \boldsymbol{w}|^2 \approx \rho$), the LoS covariance is approximated as:
\begin{align}\label{eq:app_LOS}
    \mathbb{E}\left[ (\boldsymbol{G}_{\mathtt{SI}}^{\mathtt{LoS}})^{\mathtt{H}} \boldsymbol{w} \boldsymbol{w}^{\mathtt{H}} \boldsymbol{G}_{\mathtt{SI}}^{\mathtt{LoS}} \right] &\approx \rho \boldsymbol{I}_{N_{\mathtt{r}}}.
\end{align}
Substituting \eqref{eq:NLOS_component} and \eqref{eq:app_LOS} into \eqref{eq:R_SI_decoupled}, we obtain $\boldsymbol{R}_{\mathtt{SI}}(l) \approx \chi_{\mathtt{SI}} \left( \frac{\kappa_{\mathtt{SI}} \rho \boldsymbol{I}_{N_{\mathtt{r}}}}{\kappa_{\mathtt{SI}}+1} + \frac{\rho \boldsymbol{I}_{N_{\mathtt{r}}} }{\kappa_{\mathtt{SI}}+1} \right)= \chi_{\mathtt{SI}} \rho \boldsymbol{I}_{N_{\mathtt{r}}}$. By substituting $\boldsymbol{R}_{\mathtt{SI}}(l)$ and \eqref{eq:calculate_I_up} into \eqref{eq:R_in_snapshot_decomposed}, we obtain $ \boldsymbol{R}_{\mathtt{in}}(l) = \left( I^{\mathtt{up}}(\boldsymbol{p}) + \sigma_{\mathtt{rad}}^2(\rho) \right) \boldsymbol{I}_{N_{\mathtt{r}}}$. Finally, for $L$ temporally uncorrelated snapshots, the covariance of $\boldsymbol{n}_{\mathtt{eq}}$ is given by
\begin{equation}
    \boldsymbol{R}_{\mathtt{in}} = \boldsymbol{I}_L \otimes \boldsymbol{R}_{\mathtt{in}}(l) = \left( I^{\mathtt{up}}(\boldsymbol{p}) + \sigma_{\mathtt{rad}}^2(\rho) \right) \boldsymbol{I}_{N_{\mathtt{r}} L}.
\end{equation}
\section{Proof of Remark~\ref{remark_tractability_CRLB}}\label{proof_CRLB}
Under perfect alignment $\tilde{\boldsymbol{w}} = \frac{1}{\sqrt{N_{\mathtt{t}}}}\boldsymbol{a}_{\mathtt{t}}$, substituting $\dot{\boldsymbol{G}}_\theta = \dot{\boldsymbol{a}}_{\mathtt{r},\theta} \boldsymbol{a}_{\mathtt{t}}^{\mathtt{H}} + \boldsymbol{a}_{\mathtt{r}} \dot{\boldsymbol{a}}_{\mathtt{t},\theta}^{\mathtt{H}}$ yields $ \dot{\boldsymbol{G}}_\theta \tilde{\boldsymbol{w}} = \frac{1}{\sqrt{N_{\mathtt{t}}}} \left( \dot{\boldsymbol{a}}_{\mathtt{r},\theta} (\boldsymbol{a}_{\mathtt{t}}^{\mathtt{H}} \boldsymbol{a}_{\mathtt{t}}) + \boldsymbol{a}_{\mathtt{r}} (\boldsymbol{a}_{\mathtt{t}}^{\mathtt{H}} \dot{\boldsymbol{a}}_{\mathtt{t},\theta})^{\mathtt{H}} \right)$. Under the centered array indexing convention, $\|\boldsymbol{a}_{\mathtt{t}}\|^2 = N_{\mathtt{t}}$ and the orthogonality property ($\boldsymbol{a}_{\mathtt{t}}^{\mathtt{H}} \dot{\boldsymbol{a}}_{\mathtt{t},\theta} = 0$) eliminates the second term, reducing the projections to $\dot{\boldsymbol{G}}_\theta \tilde{\boldsymbol{w}} = \sqrt{N_{\mathtt{t}}} \dot{\boldsymbol{a}}_{\mathtt{r},\theta}$ and $\dot{\boldsymbol{G}}_\phi \tilde{\boldsymbol{w}} = \sqrt{N_{\mathtt{t}}} \dot{\boldsymbol{a}}_{\mathtt{r},\phi}$. Substituting these projections into \eqref{eq:C_theta_theta}--\eqref{eq:c_phi_alpha}, the spatial metrics decouple and scale with $N_{\mathtt{t}}$, yielding \eqref{eq:c_theta_theta_final}--\eqref{eq:c_alpha_alpha_final}. Note that receive array orthogonality ($\boldsymbol{a}_{\mathtt{r}}^{\mathtt{H}} \dot{\boldsymbol{a}}_{\mathtt{r},\theta} = 0$) yields $\tilde{\boldsymbol{c}}_{\theta\bar{\boldsymbol{\alpha}}} = \tilde{\boldsymbol{c}}_{\phi\bar{\boldsymbol{\alpha}}} = \boldsymbol{0}$, fully decoupling the spatial angles from $\bar{\boldsymbol{\alpha}}$.

Substituting \eqref{eq:c_theta_theta_final}--\eqref{eq:c_alpha_alpha_final} into $\mathcal{G}_\theta(\tilde{\boldsymbol{w}})$ defined in \eqref{eq:steering_vector_theta}, its inverse simplifies to $(\mathcal{G}_\theta)^{-1} = \tilde{c}_{\theta\theta} - \frac{|\tilde{c}_{\theta\phi}|^2}{\tilde{c}_{\phi\phi}} = |\alpha|^2 N_{\mathtt{t}} U_{\mathtt{r},\theta}^{\mathtt{eff}}$, where $U_{\mathtt{r},\theta}^{\mathtt{eff}} \triangleq \frac{U_{\mathtt{r},\theta} U_{\mathtt{r},\phi} - U_{\mathtt{r},\theta\phi}^2}{U_{\mathtt{r},\phi}}$. Letting $A \triangleq \frac{N_{\mathtt{r}} \pi^2}{12}$, expanding the numerator of $U_{\mathtt{r},\theta}^{\mathtt{eff}}$ makes the cross-terms vanish: $U_{\mathtt{r},\theta} U_{\mathtt{r},\phi} - U_{\mathtt{r},\theta\phi}^2 = A^2 \cos^2(\theta) \sin^4(\phi) (N_{\mathtt{rh}}^2-1)(N_{\mathtt{rv}}^2-1)$. Inverting $(\mathcal{G}_\theta)^{-1}$ and substituting it back into \eqref{eq:CRLB_theta} yields the closed-form $\text{CRLB}_\theta$ in \eqref{eq:CRLB_final_theta}.

Similarly, substituting \eqref{eq:c_theta_theta_final}--\eqref{eq:c_alpha_alpha_final} into \eqref{eq:steering_vector_phi} yields $(\mathcal{G}_\phi)^{-1} = |\alpha|^2 N_{\mathtt{t}} \left( U_{\mathtt{r},\phi} - \frac{U_{\mathtt{r},\theta\phi}^2}{U_{\mathtt{r},\theta}} \right) = |\alpha|^2 N_{\mathtt{t}} A \sin^2(\phi) (N_{\mathtt{rv}}^2-1)$. Applying this to \eqref{eq:CRLB_phi} yields the closed-form $\text{CRLB}_\phi$ in \eqref{eq:CRLB_final_phi}.
{
\section{Proof of Remark~\ref{remark_Nt_CRLB}}\label{Proof_remark_Nt_CRLB}
Defining $\Upsilon \triangleq \frac{6 ( I^{\mathtt{up}}(\boldsymbol{p}) + \sigma_{\mathtt{rad}}^2(\rho) )}{L \rho |\alpha|^2 \pi^2}$, \eqref{eq:CRLB_final_phi} becomes $\text{CRLB}_\phi = \frac{\Upsilon}{N_{\mathtt{t}} N_{\mathtt{r}} \sin^2(\phi) (N_{\mathtt{rv}}^2-1)} = \mathcal{O}\left(\frac{1}{N_{\mathtt{t}}}\right)$. $\text{CRLB}_\theta$ shares this $\mathcal{O}(1/N_{\mathtt{t}})$ scaling, showing that transmit-array gain reduces the bounds as $N_{\mathtt{t}} \to \infty$. Second, for a square UPA ($N_{\mathtt{rh}} = N_{\mathtt{rv}} = \sqrt{N_{\mathtt{r}}}$), factoring \eqref{eq:CRLB_final_theta} and \eqref{eq:CRLB_final_phi} simplifies both bounds to $\text{CRLB}_{\{\theta,\phi\}} \propto \frac{1}{N_{\mathtt{r}}(N_{\mathtt{r}} - 1)}$. As $N_{\mathtt{r}} \gg 1$, we obtain a quadratic decay $\text{CRLB}_{\{\theta,\phi\}} = \mathcal{O}\left(1/N_{\mathtt{r}}^2\right)$.
}
\begingroup
\bibliographystyle{IEEEtran}
\bibliography{Bibliography}
\endgroup

\end{document}